\documentclass[12pt]{article}

\usepackage{a4wide}
\usepackage{amssymb}
\usepackage{amsmath,amsthm}
\usepackage[latin1]{inputenc}
\usepackage{graphicx}
\usepackage{hyperref}
\usepackage{enumerate}
\usepackage{tikz}
\usepackage{tkz-graph}
\usetikzlibrary{shapes.geometric} 
\usepackage{mathrsfs}
\usepackage{verbatim}
\usepackage{algorithm}
\usepackage[noend]{algpseudocode}
\usepackage{algorithmicx}
\usepackage{mathtools}
\usepackage{subfigure}
\usepackage{multirow}

\usepackage[normalem]{ulem}

\usepackage{array}
\newcolumntype{C}[1]{>{\centering\let\newline\\\arraybackslash\hspace{0pt}}m{#1}}

\newtheorem{definition}{Definition}
\newtheorem{theorem}{Theorem}

\hypersetup{colorlinks=true, linkcolor=blue, citecolor=blue, urlcolor=blue}

\newcommand{\iso}{\varphi}
\newcommand{\isiso}[1]{\simeq_{#1}}
\newcommand{\adv}{\mathcal{A}}
\newcommand{\dfd}{\mathcal{D}}
\newcommand{\neigh}[2]{N_{#2}(#1)}

\newcommand{\symdif}{\triangledown}
\newcommand{\dist}{\Delta}

\newcommand{\appredist}{\dist}
\newcommand{\appredistsyb}{\appredist_{syb}}
\newcommand{\appredistneigh}{\appredist_{neigh}}

\newcommand{\indsg}[2]{\langle #1 \rangle_{#2}} 
\newcommand{\wis}[2]{\langle #1 \rangle^{^{\!w}}_{#2}} 

\newcommand{\plussyb}[1]{#1^{+}}
\newcommand{\plussybg}{\plussyb{G}}
\newcommand{\pseudo}[1]{\iso#1}
\newcommand{\pseudog}{\pseudo{\plussybg}}
\newcommand{\transf}[1]{\transfName(#1)}
\newcommand{\transfg}{\transf{\pseudog}}
\newcommand{\transfName}{t}

\newcommand{\powerset}[1]{\mathcal{P}(#1)}

\newcommand{\myendeq}{\text{,}}

\newcommand{\myforall}[2]{\ensuremath{\forall_{#1 \in \{1, \ldots, #2\}}}}
\newcommand{\myrange}[2]{\ensuremath{#1 \in \{1, \ldots, #2\}}}

\newcommand{\afp}{\ensuremath{F}}
\newcommand{\threshold}{\ensuremath{\vartheta}}
\newcommand{\thresholdfp}{\ensuremath{\beta}}

\title{Robust active attacks on social graphs}

\author{Sjouke Mauw$^{1,2}$, Yunior Ram\'{i}rez-Cruz$^2$ 
and Rolando Trujillo-Rasua$^3$\\ 
{\small $^1$CSC, $^2$SnT, University of Luxembourg}\\ 
{\small 6, av. de la Fonte, L-4364 Esch-sur-Alzette, Luxembourg}\\ 
{\small $^3$School of Information Technology, Deakin University}\\ 
{\small 221 Burwood Hwy., Burwood VIC 3125, Australia}\\ 
{\small \{sjouke.mauw, yunior.ramirez\}\@@uni.lu, 
rolando.trujillo\@@deakin.edu.au}} 

\begin{document}
\maketitle

\begin{abstract} 
In order to prevent the disclosure of privacy-sensitive data, such as
names and relations between users, social network graphs have to be
anonymised before publication. Naive anonymisation of social network graphs 
often consists in deleting all identifying information of the users,
while maintaining the original graph structure. 
Various types of attacks on naively anonymised graphs have been developed.
Active attacks form a special type of such privacy attacks, in which
the adversary enrols a number of fake users, often called \emph{sybils}, 
to the social network, 
allowing the adversary to create unique structural patterns later used 
to re-identify the sybil nodes and other users
after anonymisation. Several studies have shown that adding a small
amount of noise to the published graph already suffices to mitigate
such active attacks. Consequently, active attacks have been dubbed a
negligible threat to privacy-preserving social graph publication. 
In this paper, we argue that these studies unveil shortcomings 
of specific attacks, rather than inherent problems of active attacks 
as a general strategy. 
In order to support this claim, we develop the notion of a
\emph{robust active attack}, which is an active attack that is
resilient to small perturbations of the social network graph. We
formulate the design of robust active attacks as an optimisation
problem and we give definitions of robustness for different stages of
the active attack strategy. Moreover, we introduce various heuristics
to achieve these notions of robustness and experimentally show that
the new robust attacks  are considerably more resilient than the
original ones, while remaining at the same level of feasibility. 
\end{abstract}

{\it Keywords: privacy-preserving social graph publication, active attacks}

\section{Introduction}

Data is useful. Science heavily relies on data to (in)validate 
hypotheses, discover new trends, tune up mathematical and computational models, 
etc. In other words, data collection and analysis is helping to cure 
diseases, build more efficient and environmentally-friendly buildings, take 
socially-responsible decisions,
understand our needs and those of the planet where we 
live. Despite these indisputable benefits, it is also a fact that data 
contains personal and potentially sensitive information, 
and this is where privacy and usefulness should be considered as a whole. 

A massive source of personal information is currently being handled by online 
social networks. People's life is often transparently reflected on 
popular social network platforms, such as Facebook, Twitter and Youtube. 
Therefore, releasing social network data for further study comes with a 
commitment to
ensure that users remain anonymous. Anonymity, however, is remarkably hard to 
achieve. Even a simple social graph, where an account consists of a user's 
pseudonym only and its relation to other accounts, allows users to be 
\emph{re-identified} by just considering the number of relations they
have~\cite{LT2008}. 

The use of 
pseudonyms is insufficient to guarantee anonymity. An attacker can 
cross-reference information from other sources, such as the number of 
connections, to find out the real user behind 
a pseudonym. Taking into account the type of information an attacker may have, 
called \emph{background} or \emph{prior} knowledge, is 
thus a common practice in anonymisation models. In a social graph, the 
adversary's background knowledge is regarded as any subgraph that is isomorphic 
to a subgraph in the original social graph. Various works bound 
the adversary's 
background knowledge to a specific family of graphs. For example, Liu and Terzi 
assume that the adversary's background knowledge is fully defined by star 
graphs\footnote{A star graph is a tree of depth one.}~\cite{LT2008}. This 
models an adversary that knows the degrees of the victim vertices. Others 
assume that an adversary may know the subgraphs induced by the neighbours of 
their victims~\cite{ZP2008}, and so on~\cite{ZCO2009}. 

A rather different notion of background knowledge was introduced by Backstrom 
et al.~\cite{BDK2007}. They describe an adversary able to register 
several (fake) accounts to the network, called \emph{sybil accounts}. 
The sybil accounts establish links between themselves and also with the victims. 
Therefore, in Backstrom et al.'s attack to a social graph 
$G=(V,E)$, the adversary's background knowledge is the induced 
subgraph formed by the sybil accounts in $G$ joined
with the connections to the victims. 

The adversary introduced by Backstrom et al.~\cite{BDK2007} is said to be 
\emph{active}, because he influences the structure of the social network. 
Previous authors have claimed that active attacks are either 
unfeasible or detectable. Such a claim is 
based on two observations. First, inserting many sybil nodes is hard, 
and they may be detected and removed by sybil detection techniques \cite{NS2009}. 
Second, active attacks have been reported to suffer from low resilience, 
in the sense that the attacker's ability to successfully recover the 
sybil subgraph and re-identify the victims is easily lost after a relatively 
small number of (even arbitrary) changes are introduced 
in the network \cite{Ji2015,MauwRamirezTrujillo2016,MRT2018,MauwTrujilloXuan2016}. 
As a consequence, active attacks have been largely overlooked in literature. 
Backstrom et al. argue for the feasibility of active attacks, showing that 
proportionally few sybil nodes (in the order of $\log n$ nodes for networks 
of order $n$) are sufficient for compromising any legitimate node. 
This feature of active attacks is relevant in view of the fact that sybil 
defence mechanisms 
do not attempt to remove every sybil node, but to limit their number 
to no more than $\log n$ \cite{Yu2006,Yu2008}, which entails that 
sufficiently capable sybil subgraphs are likely to go unnoticed 
by sybil defences. The second claim, that of lack of resiliency 
to noisy releases, is the main focus of this work. 
\vspace*{0.2cm}

\noindent \textbf{Contributions.}
In this paper we show that active attacks do constitute a serious threat for 
privacy-preserving publication of social graphs. We do so by proposing the 
first active attack strategy that features two key properties. 
Firstly, it can effectively re-identify users with a small number 
of sybil accounts. Secondly, it is resilient, in the sense that it resists 
not only the introduction of reasonable amounts of noise in the network, 
but also the application of anonymisation algorithms specifically  
designed to counteract active attacks. 
The new attack strategy is based on new notions of robustness for the sybil 
subgraph and the set of fingerprints, as well as noise-tolerant algorithms 
for sybil subgraph retrieval and re-identification. 
The comparison of the robust active attack strategy to the original 
active attack is facilitated by the introduction of a novel framework of study, 
which views an active attack as an attacker-defender game. 

The remainder of this paper is structured as follows. 
We enunciate our adversarial model in the form of an attacker-defender 
game in Section~\ref{sec-adversarial-model}. Then, the new notions of robustness 
are introduced in Section~\ref{sec-ideal-robust-attack}, 
and their implementation is discussed 
in Section~\ref{sec-robustisation-walk-based}. 
Finally, we experimentally evaluate our proposal in Section~\ref{sec-experiments}, 
discuss related work in Section~\ref{sec-related-work},  
and give our conclusions in Section~\ref{sec-conclusions}.

\section{Adversarial model}\label{sec-adversarial-model}

We design a game between an attacker $\adv$ and a defender $\dfd$.
The goal of the attacker is to identify the victim nodes after
pseudonymisation and transformation of the graph by the defender. 
We first introduce the necessary graph theoretical notation, and then
formulate the three stages of the attacker-defender game. 

\subsection{Notation and terminology}\label{sec-notation}

We use the following standard notation and terminology.
Additional notation that may be needed in other sections of the paper 
will be introduced as needed. 

\begin{itemize}

\item
A \emph{graph} $G$ is represented as a pair $(V,E)$, where $V$ is a set
of vertices (also called nodes) and $E\subseteq V\times V$ is a set of edges.
The vertices of $G$ are denoted by $V_G$ and its edges by $E_G$.
As we will only consider undirected graphs, we will consider an edge
$(v,w)$ as an unordered pair. We will use the notation $\mathcal{G}$ 
for the set of all graphs. 

\item
An \emph{isomorphism} between two graphs $G=(V,E)$ and $G'=(V',E')$ is
a bijective function $\iso\colon V \to V'$, such that
$\forall_{v1,v2\in V} . (v_1,v_2)\in E \iff (\iso(v_1),\iso(v_2))\in E'$.
Two graphs are \emph{isomorphic}, denoted by $G\isiso{\iso} G'$, or
briefly $G\isiso{} G'$, if there exists an isomorphism $\iso$ between them. 
Given a subset of vertices $S \subseteq V$, we will often use $\iso S$ to 
denote the set $\{\iso(v) | v \in S\}$.

\item
The set of \emph{neighbours} of a set of nodes
$W\subseteq V$ is defined by
$\neigh{W}{G} = \{v\in V\setminus W \mid \exists_{w\in W} . (v,w)\in
E \lor (w,v)\in E\}$.
If $W=\{w\}$ is a singleton set, we will write
$\neigh{w}{G}$ for $\neigh{\{w\}}{G}$. 
The \emph{degree} of a vertex $v\in V$, denoted as $\delta_G(v)$, is 
defined as $\delta_G(v)=|\neigh{v}{G}|$. 

\item
Let $G=(V,E)$ be a graph and let $S\subseteq V$.
The \emph{weakly-induced subgraph} of $S$ in $G$, denoted by
$\wis{S}{G}$, is the subgraph of $G$ with vertices $S\cup\neigh{S}{G}$ and
edges $\{(v,v')\in E \mid v\in S \lor v'\in S\})$. 
\end{itemize}

\subsection{The attacker-defender game}

The attacker-defender game starts with a graph $G=(V,E)$ representing a 
snapshot of a social network. The attacker knows a subset of the users, 
but not the connections between them. An example is shown 
in Figure~\ref{fig-active-attack}(a), where capital letters represent 
the real identities of the users and dotted lines represent 
the relations existing between them, which are not known to the adversary. 
Before a pseudonymised graph 
is released, the attacker manages to enrol sybil 
accounts in the network and establish links with the victims, as depicted 
in Figure~\ref{fig-active-attack}(b), where sybil accounts 
are represented by dark-coloured nodes and the edges known to the adversary 
(because they were created by her) are represented by solid lines. 
The goal of the attacker is to later re-identify the victims in order to learn 
information about them. 
The defender anonymises the social graph by removing the real user 
identities, or replacing them with pseudonyms, and possibly perturbing 
the graph. In Figure~\ref{fig-active-attack}(c) we illustrate the 
pseudonymisation process of the graph in Figure~\ref{fig-active-attack}(b). 
The pseudonymised graph contains information that the attacker wishes 
to know, such as the existence of relations between users, 
but the adversary cannot directly learn this information, as the identities 
of the vertices are hidden. Thus, after the pseudonymised graph is published, 
the attacker analyses the graph to first re-identify its own sybil accounts, 
and then the victims (see Figure~\ref{fig-active-attack}(d)). 
This allows her to acquire new information, which was supposed 
to remain private, such as the fact that $E$ and $F$ are friends.  

\begin{figure*}[htb!]
\newcommand{\scaleValue}{1.5}
\newcommand{\spaceValue}{1.5cm}
\centering
	\subfigure[A fragment of the original 
	graph. ]{\includegraphics[scale=\scaleValue]{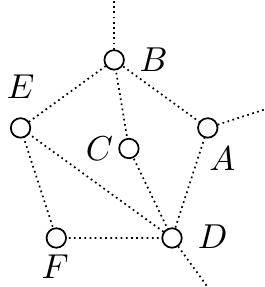}}
	\hspace{\spaceValue}
	\subfigure[The sybil nodes are added and fingerprints 
	are created for the victims.	\label{fig-sybils}]
	{\includegraphics[scale=\scaleValue]{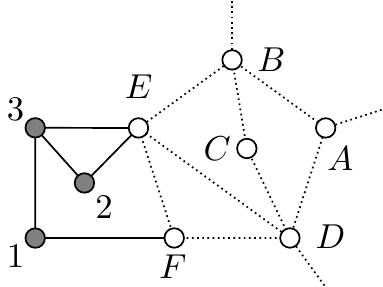}}\hspace{\spaceValue}
	\subfigure[Pseudonymised graph after 
	release.]{\includegraphics[scale=\scaleValue]{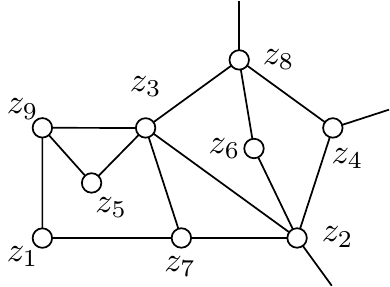}}
	\hspace{\spaceValue}
	\subfigure[The attacker subgraph is recovered and the 
	victims are re-identified. 
	\label{fig-reidentification}]{\includegraphics[scale=\scaleValue]{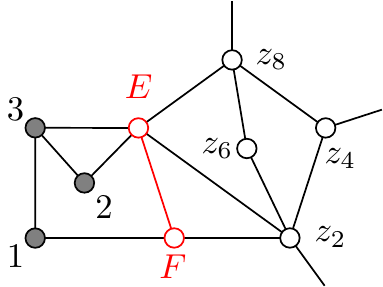}} 	 
\caption{The attacker-defender game.} 
\label{fig-active-attack}
\end{figure*}

Next, we formalise the three stages of the attacker-defender game, 
assuming a graph $G=(V,E)$. 

\begin{enumerate}

\item
\emph{Attacker subgraph creation.} 
The attacker constructs a set of \emph{sybil} nodes 
$S=\{x_1,x_2,$ $\ldots,x_{|S|}\}$, such 
that $S\cap V = \emptyset$ and a set of edges
$F\subseteq (S\times S) \cup (S\times V) \cup (V\times S)$. 
It clearly follows that $E\cap F = \emptyset$. 
We call $\plussybg =(V\cup S, E\cup F)$ the \emph{sybil-extended}
graph of $G$.
The attacker does not know the complete graph $\plussybg$, but he knows
$\wis{S}{\plussybg}$, the weakly-induced subgraph of $S$ in $\plussybg$. 
We say that $\wis{S}{\plussybg}$ is the \emph{attacker subgraph}. 
The attacker subgraph creation has two substages:
\begin{enumerate}
\item \emph{Creation of inter-sybil connections.} 
A unique (with high probability) and efficiently retrievable connection pattern 
is created between sybil nodes to facilitate the attacker's 
task of retrieving the sybil subgraph at the final stage. 
\item
\emph{Fingerprint creation.} 
For a given \emph{victim} vertex 
$y \in N_{\plussybg}(S) \setminus S$, 
we call the victim's neighbours in $S$, i.e 
$N_{\plussybg}(y)\cap S$, its \emph{fingerprint}. 
Considering the set of victim vertices $Y = \{y_1, \ldots, y_m\}$, 
the attacker ensures that $N_{\plussybg}(y_i)\cap S \neq 
N_{\plussybg}(y_j)\cap S$ 
for every $y_i,y_j\in Y$, $i\neq j$. 
\end{enumerate} 

\item
\emph{Anonymisation.} 
The defender obtains $\plussybg$ and constructs an isomorphism $\iso$
from $\plussybg$ to $\pseudog$. We call $\pseudog$ the
\emph{pseudonymised} graph. The purpose of pseudonymisation is to remove 
all personally identifiable information from the vertices of $G$. 
Next, given a non-deterministic procedure $\transfName$ 
that maps graphs to graphs, known by both $\adv$ and $\dfd$, the defender 
applies transformation $\transfName$ to $\pseudog$, 
resulting in the \emph{transformed} graph $\transfg$. 
The procedure $t$ modifies $\pseudog$ by adding and/or removing vertices 
and/or edges. 
\item
\emph{Re-identification.} After obtaining $\transfg$, the attacker executes the 
re-identification attack in two stages.
\begin{enumerate}
	\item \emph{Attacker subgraph retrieval.} Determine the isomorphism $\iso$ 
	restricted to the domain of sybil nodes $S$. 
	\item \emph{Fingerprint matching.} Determine the isomorphism $\iso$ 
	restricted to the domain of victim 
	nodes $\{y_1,y_2,\ldots,y_m\}$. 
\end{enumerate}
\end{enumerate}

As established by the last step of the attacker-defender game, we consider the 
adversary to succeed if she effectively determines the 
isomorphism $\iso$ restricted to the domain of victim nodes 
$\{y_1,y_2,\ldots,y_m\}$. That is, when the adversary re-identifies all victims 
in the anonymised graph. 

\section{Robust active attacks}\label{sec-ideal-robust-attack}

This section formalises robust active attacks. We provide 
mathematical formulations, in the form of optimisation problems, 
of the attacker's goals in the first and third stages. In particular, 
we address three of the subtasks that need to be accomplished in these 
stages: fingerprint creation, attacker subgraph retrieval and fingerprint matching. 

\subsection{Robust fingerprint creation}\label{sec-ideal-fp-creation}

Active attacks, in their original formulation \cite{BDK2007}, 
aimed at re-identifying victims in pseudo\-nym\-ised graphs. 
Consequently, the uniqueness of every fingerprint 
was sufficient to guarantee success with high probability, provided 
that the attacker subgraph is correctly retrieved. 
Moreover, as shown in \cite{BDK2007}, several types of randomly generated 
attacker subgraphs can indeed be correctly and efficiently retrieved, 
with high probability, after pseudonymisation. 
The low resilience reported for this approach when the pseudonymised graph 
is perturbed by applying an anonymisation 
method \cite{MauwRamirezTrujillo2016,MRT2018,MauwTrujilloXuan2016} 
or by introducing arbitrary changes \cite{Ji2015}, 
comes from the fact that it relies on finding 
exact matches between the fingerprints created by the attacker 
at the first stage of the attack and their images in $\transfg$. 
The attacker's ability to find such exact matches is lost even after 
a relatively small number of perturbations is introduced by $t$. 

Our observation is that setting for the attacker the goal of obtaining 
the exact same fingerprints in the perturbed graph is not only too strong, 
but more importantly, not necessary. Instead, we 
argue that it is sufficient for the attacker to
obtain a set of fingerprints that is 
close enough to the original set of fingerprints, for some notion of
closeness. Given 
that a fingerprint is a set of vertices, we propose to use the
cardinality of the 
symmetric difference of two sets to measure the distance between
fingerprints. The 
symmetric difference between two sets $X$ and 
$Y$, denoted by $X \symdif Y$, is the set of elements in $X \cup Y$ that are not 
in $X \cap Y$. We use $d(X, Y)$ to denote $|X \symdif Y|$.

Our goal at this stage of the attack is to create a set of fingerprints 
satisfying the following property. 

\begin{definition}[Robust set of fingerprints]\label{def-fp-creation-ideal}
Given a set of victims $\{y_1, 
\ldots, y_m \}$ and a set of sybil nodes $S$ in a graph $\plussybg$, the set 
of fingerprints $\{\afp_1, \ldots, \afp_m\}$ with $\afp_i = 
N_{\plussybg}(y_i)\cap S$ is 
said to be \emph{robust} if it maximises 
\begin{equation}\label{eq-obj-fn-robust-fp}
\min_{1 \leq i < j \leq m} \{d(\afp_i, \afp_j)\}\text{.}
\end{equation}
\end{definition}

The property above ensures that the lower bound on the distance between any 
pair of fingerprints is maximal. In what follows, we will refer 
to the lower bound defined by Equation~(\ref{eq-obj-fn-robust-fp}) 
as \emph{minimum separation} of a set of fingerprints. 
For example, in Figure~\ref{fig-active-attack}(b), the fingerprint of the 
vertex $E$ with respect to the set of attacker vertices $\{1, 2, 
3\}$ is $\{2, 3\}$, and the fingerprint of the vertex $F$ is 
$\{1\}$. This gives a minimum 
separation between the two victim's fingerprints equal to $| \{2, 3\} \symdif 
\{1\} | = | \{1, 2, 3\} | = 3$, which is maximum.
Therefore, given attacker vertices $\{1,2,3\}$, the set of
fingerprints $\{\{2,3\},\{1\}\}$ is robust for the set of victim nodes
$\{E,F\}$.

Next we prove that, if the distance between each original fingerprint $F$ 
and the corresponding 
anonymised fingerprint $\iso F$ is less than half the minimum separation, 
then the distance between $F$ and any other anonymised 
fingerprint, say $\iso F'$, is strictly larger than half the minimum 
separation. 

\begin{theorem}\label{theo-separation}
Let $S$ be the set of sybil nodes, let $\{y_1, 
\ldots, y_m \}$ be the set of victims and let
$\{\afp_1, \ldots, \afp_m\}$ be their fingerprints with minimum separation 
$\delta$. 
Let  $\afp_i'$ be the fingerprint of $\iso y_i$ in the anonymised graph 
$\transfg$, for $\myrange{i}{m}$. Then,

\[
\myforall{i}{m} d(\iso F_i, F_i') < \delta/2 \implies \myforall{i, j}{m} i \neq 
j 
\implies d(\iso F_i, 
F_j') > \delta/2 \myendeq
\]

\end{theorem}

\begin{proof}
In order to reach a contradiction, we assume that $d(\iso F_i, F_j')
\leq \delta/2$ for some $i, j \in \{1, \ldots, m\}$ 
with 
$i \neq j$. Because $d(\iso F_j, F_j') < \delta/2$, we have $d(\iso F_i, F_j') +
d(\iso F_j, F_j') < \delta$. By the triangle inequality we obtain that 
$d(\iso F_i, \iso F_j 
) < d(\iso F_i, F_j') + d(\iso F_j, F_j') < \delta$. Hence $d(\iso F_i, \iso 
F_j)$ is lower than 
the minimum separation of $\{\afp_1, \ldots, \afp_m\}$, which yields a
contradiction given that $d(\iso F_i, \iso 
F_j) = d(F_i, F_j) \geq \delta$. 
\end{proof}

We exploit Theorem~\ref{theo-separation} later in the 
fingerprint matching step through the following corollary. 
If $\delta/2$ is the maximum distance shift from an original 
fingerprint $F_i$ of $y_i$ to the fingerprint $F_i'$ of $y_i$ in the perturbed 
graph, then for every $\afp \in \{\afp_1', \ldots, \afp_m'\}$ it holds that 
$d(\afp, \iso \afp_i) < \delta/2 \iff \afp = \iso \afp_i$. 
In other words, given a set of victims for which a set of fingerprints needs 
to be defined, the larger the minimum separation of these fingerprints, 
the larger the number of perturbations that can be tolerated 
in $\transfg$, while still being able to match the perturbed fingerprints 
to their correct counterparts in $\plussybg$. 

As illustrated earlier in our running example, the fingerprints of $E$ and $F$ 
are $\{2, 3\}$ and $\{1\}$, respectively, which gives a minimum separation 
of $\delta = 3$. Theorem~\ref{theo-separation} states that, if after 
anonymisation 
of the graph, the fingerprints of $E$ and $F$ become, say $\{2\}$ and $\{1, 
2\}$, 
respectively, then it must hold that $|\{2, 3\} \symdif \{1, 2\}| > 3/2$ and 
$|\{1\} \symdif 
\{2\}| > 3/2$, while $|\{2, 3\} \symdif \{2\}| < 3/2$ and $|\{1\} \symdif 
\{1, 2\}| < 3/2$. This makes it easy to match the original fingerprint, say 
$\{2, 3\}$, with the correct perturbed fingerprint $\{2\}$ by calculating their 
distance and verifying that it remains below the threshold $\delta/2$. 
In Subsection~\ref{sec-ssgr-fp-creation}, we will describe an efficient 
algorithm for addressing this optimisation problem. 

\subsection{Robust attacker subgraph retrieval}
\label{subsec-syb-retr-ideal}

Let $\mathcal{C}=\{\wis{X}{\transfg} \mid X\subseteq V_{\transfg}, 
|X|=|S|, \wis{X}{\transfg}\cong \wis{S}{\plussybg}\}$ be the set 
of all subgraphs isomorphic to the attacker subgraph 
$\wis{S}{\plussybg}$ and weakly induced in $\transfg$ by a vertex subset 
of cardinality $|S|$. The original active attack formulation \cite{BDK2007}
assumes  
that $|\mathcal{C}|=1$ and that the subgraph in $\mathcal{C}$ is the 
image of the attacker 
subgraph after pseudonymisation. This assumption, for example, holds on 
the pseudonymised graph in Figure~\ref{fig-active-attack}(c). But it 
hardly holds on perturbed graphs as demonstrated 
in~\cite{Ji2015,MauwRamirezTrujillo2016,MRT2018,MauwTrujilloXuan2016} . In 
fact, 
$\mathcal{C}$ becomes empty 
by simply adding an edge between any pair of attacker nodes, which makes the 
attack fail quickly when increasing the amount of perturbation. 

To account for the occurrence of perturbations in releasing $\transfg$, 
we introduce the notion of robust attacker subgraph 
retrieval. Rather than limiting the retrieval process to finding 
an exact match of the original attacker subgraph, we consider that 
it is enough to find a sufficiently similar subgraph, thus adding some 
level of noise-tolerance. By ``sufficiently similar", we mean a graph 
that minimises some graph dissimilarity measure 
$\dist\colon\mathcal{G}\times \mathcal{G}\rightarrow\mathbb{R}^+$ with respect 
to $\wis{S}{\plussybg}$. The problem is formulated as follows. 
 
\begin{definition}[Robust attacker subgraph retrieval problem]
\label{def-ideal-sybil-subgraph-retrieval} 
Given a graph dissimilarity measure 
$\dist\colon\mathcal{G}\times \mathcal{G}\rightarrow\mathbb{R}^+$, 
and a set $S$ of sybil nodes in the graph $\plussybg$, 
find a set $S' \subseteq V_{\transfg}$ that minimises 
\begin{equation}\label{eq-obj-fn-robust-syb-subgr}
\dist(\wis{S'}{\transfg}, \wis{S}{\plussybg}).
\end{equation}
\end{definition}
A number of graph (dis)similarity measures have been proposed 
in the literature \cite{SF1983,BDK2007,Bunke2000,Fober2013,Mallek2015}.  
Commonly, the choice of a particular measure is \emph{ad hoc}, 
and depends on the characteristics of the graphs being compared. 
In Subsection~\ref{sec-robust-retrieval}, we will describe 
a measure that is efficiently computable 
and exploits the known structure of $\wis{S}{\plussybg}$, 
by separately accounting for inter-sybil and sybil-to-non-sybil edges. 
Along with this dissimilarity measure, we provide an algorithm 
for constructively finding a solution to the problem enunciated 
in Definition~\ref{def-ideal-sybil-subgraph-retrieval}.  

\subsection{Robust fingerprint matching}

As established by the attacker-defender game discussed in 
Section~\ref{sec-adversarial-model}, fingerprint matching 
is the last stage of the active attack. Because it clearly relies on the 
success of the previous steps, we make the following two assumptions 
upfront.

\begin{enumerate}
	\item We assume that the robust sybil subgraph  
	retrieval procedure succeeds, i.e. that $\iso S = S'$ where $S'$ is the set 
	of 
	sybil nodes obtained in the previous step. 
	\item  Given the original set 
of victims $Y$, we assume that the 
	set of vertices in the neighbourhood of $S'$ 
	contains those in $\iso Y$, i.e. $\iso Y \subseteq 
N_{\transfg}(S')\setminus S'$, otherwise $S'$ is insufficient information to 
achieve the 
goal of re-identifying all victim vertices. 

\end{enumerate}

Given a potentially correct set of sybil nodes $S'$ and a set of potential 
victims $Y'=\{y_1', \ldots, y_n' \} = N_{\transfg}(S')\setminus S'$, 
the re-identification process consists in determining the isomorphism $\iso$ 
restricted to the vertices in $Y'$. Next we define re-identification as an 
optimisation problem, and after that we provide sufficient conditions under 
which a solution leads to correct identification.

\begin{definition}[Robust re-identification 
problem]\label{def-fp-match-ideal}
Let $S$ and $S'$ be the set of sybil nodes in the original and anonymised 
graph, respectively. Let $\{y_1, \ldots, y_m \}$ be the victims in $\plussybg$ 
with fingerprints $F_1 = N_{\plussybg}(y_1)\cap S, \ldots, F_m = 
N_{\plussybg}(y_m)\cap S$. 
The \emph{robust 
re-identification problem} consists in finding an isomorphism
$\phi\colon S \rightarrow 
S'$ and subset 
$\{z_1, \ldots, z_m\} \subseteq N_{\transfg}(S')\setminus S'$ that minimises
\begin{equation}\label{eq-obj-fn-robust-reident}
\lVert (d(\phi F_1, N_{\transfg}(z_1)\cap S'), \ldots, d(\phi F_m, 
N_{\transfg}(z_m)\cap S')) \rVert_{\infty}.
\end{equation}
where $\lVert . \rVert_{\infty}$ stands for the infinity norm.
\end{definition}

Optimising the infinity norm gives the 
lowest upper bound on the distance between 
an original fingerprint and the fingerprint of a vertex in the perturbed graph. 
This is useful towards the goal of correctly re-identifying all victims. 
However, should the adversary aims at re-identifying at least one victim with 
high probability, then other plausible objective functions can be used, such as 
the Euclidean norm.

As stated earlier, our intention is to exploit 
the result of Theorem~\ref{theo-separation}, 
provided that the distance between original and perturbed fingerprints is lower 
than $\delta/2$, where $\delta$ is the minimum separation of the original set 
of fingerprints. This is one out of three conditions that we prove sufficient 
to infer a correct mapping  
$\iso$ from a solution to the robust re-identification problem, 
as stated in the following result. 

\begin{theorem}\label{theo-conds}
Let $\phi\colon S \rightarrow S'$ and $\{z_1, \ldots, z_m\}$ be a solution to 
the robust re-identification problem defined by  the set of 
sybil nodes $S$ in the original graph $\plussybg$, the set of 
sybil nodes $S'$ in the anonymised 
graph $\transfg$, and the set of victims $\{y_1, \ldots, y_m \}$ in 
$\plussybg$. Let $\{F_1, \ldots, F_m\}$ be the set of fingerprints of $\{y_1, 
\ldots, y_m\}$ and $\delta$ its minimum separation. If the following three 
conditions hold:
\begin{enumerate}	
	\item $\phi$ is a subset of $\iso$, i.e. $\forall_{x \in S} \phi(x) = 
	\iso(x)$
	\item $\{\iso(y_1), \ldots, \iso(y_m)\} = 
	N_{\transfg}(S')\setminus S'$ 
	\item For every $y_i \in \{y_1, \ldots, y_m \}$, $d(\iso \afp_i, \afp_i')
	< 	\delta/2$ where $F_i' = N_{\transfg}(\iso(y_i))\cap S'$,
\end{enumerate}
then $\iso(y_i) = z_i$ for every $\myrange{i}{m}$.
\end{theorem}

\begin{proof}
From the third condition we obtain that the correct mapping $\iso$ satisfies, 

\[
\max\{d(\iso F_1, F_1'), \ldots, d(\iso F_m, 
F_m')\} < \delta/2 \myendeq
\]

Now, the second condition gives that $\{\iso(y_1), \ldots, 
\iso(y_m)\} = \{z_1, \ldots, z_m\}$. This means that, for every  $i 
\in \{1, \ldots, m\}$, $F_i' = N_{\transfg}(z_j)\cap S'$ for some $j 
\in \{1, \ldots, m\}$. Let $f$ be an automorphism in $\{1, \ldots, m\}$ such 
that $F_i' = 
N_{\transfg}(z_{f(i)})\cap S'$ for every $i \in \{1, \ldots, m\}$. We use 
$f^{-1}$ to denote the inverse of $f$. Then, considering that $\phi F_i = \iso 
F_i$ for every $i 
\in \{1, 
\ldots, m\}$ (first condition), we obtain the following 
equalities.

\begin{align*}
& \max\{d(\phi F_1, N_{\transfg}(z_1)\cap S'), \ldots, d(\phi F_m, 
N_{\transfg}(z_m)\cap S')\} = \\
& \max\{d(\iso F_1, N_{\transfg}(z_1)\cap S'), \ldots, d(\iso F_m, 
N_{\transfg}(z_m)\cap S')\} = \\
& \max\{d(\iso F_1, F'_{f^{-1}(1)}), \ldots, d(\iso F_m, 
F'_{f^{-1}(m)})\}
\end{align*}

Considering Theorem~\ref{theo-separation}, we obtain that for every $i, j 
\in \{1, \ldots, m\}$ with $i \neq j$ 
it holds that $d(\iso F_i, F_j') > \delta/2$. Therefore, if $f$ is not the 
trivial automorphism, i.e. $f(i) = i \; \forall i \in \{1, \ldots, m\}$, then 
$\max\{d(\iso F_1, F'_{f^{-1}(1)}), \ldots, d(\iso F_m, 
F'_{f^{-1}(m)})\} > \delta/2$. This implies that, 

\[
\max\{d(\phi F_1, N_{\transfg}(z_1)\cap S'), \ldots, d(\phi F_m, 
N_{\transfg}(z_m)\cap S')\} > \delta/2 \myendeq
\]

However, this contradicts the optimality of the solution $\phi$ and $\{z_1, 
\ldots, z_m\}$. Therefore, $f$ must be the trivial automorphism, which 
concludes the proof.
\end{proof}

In Theorem~\ref{theo-conds}, the first condition states that the adversary 
succeeded on correctly identifying each of her own sybil nodes in the perturbed 
graph. That is to say, the adversary retrieved the mapping $\iso$ 
restricted to 
the set of victims. This is clearly an important milestone in the attack as 
victim's fingerprints are based on such mapping. The second condition says that 
the neighbours of the sybil vertices remained the same after perturbation. As a 
result, the adversary knows that $\{z_1, \ldots, z_m\}$ is the victim set in 
the perturbed graph, but 
she does not know yet the 
isomorphism $\iso$ restricted to the set of victims $\{y_1, \ldots, y_m\}$.
Lastly, the third condition states that $\delta/2$ is an upper bound on the 
distance between a 
victim's fingerprints in the pseudonymised graph $\pseudog$ and the perturbed 
graph $\transfg$, where $\delta$ is the the minimum 
separation between the victim's fingerprints. In other words, the 
transformation method did not perturbed a victim's fingerprint ``too much''. 
If those three conditions hold, Theorem~\ref{theo-conds} shows that the 
isomorphism $\iso$ restricted to the set of victims $\{y_1, \ldots, y_m\}$ is 
the trivial isomorphism onto $\{z_1, \ldots, z_m\}$.

\bigskip
\noindent
\textbf{Summing up:} in this section we have enunciated the three problem 
formulations for robust active attacks, namely: 
\begin{itemize}
\item Creating a robust set of fingerprints. 
\item Robustly retrieving the attacker subgraph in the perturbed graph. 
\item Robustly matching the original fingerprints to perturbed fingerprints. 
\end{itemize}  
Additionally, we have defined a set of conditions under which finding 
a solution for these problems guarantees a robust active attack to be successful. 
Each of the three enunciated problem has been stated as an optimisation task. 
Since obtaining exact solutions to these problems is computationally expensive, 
in the next section we introduce heuristics for finding approximate solutions.  

\section{Heuristics for an approximate instance 
of the robust active attack strategy} 
\label{sec-robustisation-walk-based} 

In this section we present the techniques for creating an instance 
of the robust active attack strategy described in the previous section. 
Since finding exact solutions to the optimisation problems 
in Equations~(\ref{eq-obj-fn-robust-fp}), (\ref{eq-obj-fn-robust-syb-subgr}) 
and~(\ref{eq-obj-fn-robust-reident}) is computationally expensive, 
we provide efficient approximate heuristics. 

\subsection{Attacker subgraph creation}\label{sec-ssgr-fp-creation} 

For creating the internal links of the sybil subgraph, we will use 
the same strategy as the so-called \emph{walk-based attack} \cite{BDK2007}, 
which is the most widely-studied instance of the original active attack strategy. 
By doing so, we make our new 
attack as (un)likely as the original to have the set of sybil nodes removed 
by sybil defences. 
Thus, for the set of sybil nodes $S$, the attack will set an arbitrary  
(but fixed) order among the elements of $S$. Let $x_1,x_2,\ldots,x_{|S|}$ 
represent the vertices of $S$ in that order. The attack will firstly create 
the path $x_1x_2\ldots x_{|S|}$, whereas the remaining inter-sybil edges 
are independently created with probability~$0.5$. 

For creating the set of fingerprints, we will apply a greedy algorithm 
for maximising the minimum separation defined 
in Equation~(\ref{eq-obj-fn-robust-fp}). The idea behind the algorithm 
is to arrange all possible fingerprints in a grid-like auxiliary graph, 
in such a way that nodes representing similar 
fingerprints are linked by an edge, and nodes representing well-separated 
fingerprints are not. Looking for a set of maximally separated fingerprints 
in this graph reduces to a well-known problem in graph theory, 
namely that of finding an independent set. An \emph{independent set} $I$ 
of a graph $G$ is a subset of vertices 
from $G$ such that $E_{\indsg{I}{G}}=\emptyset$, that is, all vertices in $I$ 
are pairwise not linked by edges. 
If the graph is constructed in such a way that every pair of fingerprints 
whose distance is less then or equal to some value $i$, then an independent set 
represents a set of fingerprints with a guaranteed minimum separation 
of at least $i+1$. For example, the fingerprint graph shown 
in Figure~\ref{fig-fp-graph} (a) represents the set of fingerprints 
$\{\{1\},\{2\},\{3\},\{1,2\},\{1,3\},\{2,3\},\{1,2,3\}\}$, 
with edges linking all pairs $X,Y$ of fingerprints such that $d(X,Y)\le 1$, 
whereas Figure~\ref{fig-fp-graph} (b) represents an analogous graph 
where edges link all pairs $X,Y$ of fingerprints such that $d(X,Y)\le 2$. 
Note that the vertex set of both graphs is the power set of $\{1,2,3\}$, 
except for the empty set, which does not represent a valid fingerprint, 
as every victim must be linked to at least one sybil node. 
A set of fingerprints built from an independent set of the first graph 
may have minimum separation~$2$ (e.g. $\{\{1\},\{2\},\{1,2,3\}\}$) 
or $3$ (e.g. $\{\{1,3\},\{2\}\}$), whereas a set of fingerprints built 
from an independent set of the second graph will have minimum separation $3$ 
(the independent sets of this graph are $\{\{1\},\{2,3\}\}$, $\{\{1,2\},\{3\}\}$ 
and $\{\{1,3\},\{2\}\}$). 

\begin{figure}[!ht]

\begin{center}

\begin{tikzpicture}[inner sep=0.7mm, place/.style={circle,draw=black,
fill=white},xx/.style={circle,draw=black!99, 
fill=black!99},gray1/.style={circle,draw=black!99, 
fill=black!25},gray2/.style={circle,draw=black!99, 
fill=black!50},gray3/.style={circle,draw=black!99,fill=black!75},
gray2d/.style={diamond,draw=black!99,fill=black!50},
transition/.style={rectangle,draw=black!50,fill=black!20,thick}, 
line width=.5pt]

\def\radius{1.5cm}
\def\lblradius{2cm}

\def\n{6}

\foreach \ind in {1,...,\n}\pgfmathparse{90+360/\n*\ind}\coordinate (g1\ind) at 
(\pgfmathresult:\radius);

\foreach \ind in {1,...,\n}\pgfmathparse{90+360/\n*\ind}\coordinate (v1\ind) at 
(\pgfmathresult:\lblradius);

\coordinate (cv1) at (0,0);

\draw[black] (g11) -- (g12) -- (g13) -- (g14) -- (g15) -- (g16) -- cycle;
\draw[black] (cv1) -- (g11);
\draw[black] (cv1) -- (g13);
\draw[black] (cv1) -- (g15);

\foreach \ind in {1,...,\n} \node [place] at (g1\ind) {};
\node [place] at (cv1) {};

\node at (v11) {$\{1,3\}$};
\node at (v12) {$\{1\}$};
\node at (v13) {$\{1,2\}$};
\node at (v14) {$\{2\}$};
\node at (v15) {$\{2,3\}$};
\node at (v16) {$\{3\}$};
\coordinate [label=center:{\small{$\{1,2,3\}$}}] (lv1) at (0.67,-0.27);

\coordinate [label=center:{\small{(a)}}] (al) at (0,-2.5);

\foreach \ind in {1,...,\n}\pgfmathparse{90+360/\n*\ind}\coordinate[xshift=5cm] (g2\ind) at 
(\pgfmathresult:\radius);

\foreach \ind in {1,...,\n}\pgfmathparse{90+360/\n*\ind}\coordinate[xshift=5cm] (v2\ind) at 
(\pgfmathresult:\lblradius);

\coordinate (cv2) at (5,0);

\draw[black] (g21) -- (g22) -- (g23) -- (g24) -- (g25) -- (g26) -- cycle;
\draw[black] (g21) -- (g23) -- (g25) -- cycle;
\draw[black] (g22) -- (g24) -- (g26) -- cycle;
\foreach \ind in {1,...,\n} \draw[black] (cv2) -- (g2\ind);

\foreach \ind in {1,...,\n} \node [place] at (g2\ind) {};
\node [place] at (cv2) {};

\node at (v21) {$\{1,3\}$};
\node at (v22) {$\{1\}$};
\node at (v23) {$\{1,2\}$};
\node at (v24) {$\{2\}$};
\node at (v25) {$\{2,3\}$};
\node at (v26) {$\{3\}$};
\coordinate [label=center:{\tiny{$_{\{1,2,3\}}$}}] (lv2) at (5.51,-0);

\coordinate [label=center:{\small{(b)}}] (al) at (5,-2.5);

\end{tikzpicture}

\end{center}

\caption{The fingerprint graphs $(\powerset{\{1,2,3\}}\setminus \{\emptyset\}, 
\{(X,Y)\;\mid\; X,Y\in \powerset{\{1,2,3\}}\setminus \{\emptyset\}, X\neq Y, 
d(X,Y)\le i\})$ for (a) $i=1$ and (b) $i=2$.} 
\label{fig-fp-graph} 
\end{figure} 

Our fingerprint generation method iteratively creates increasingly 
denser fingerprint graphs. The vertex set of every graph is the set of possible 
fingerprints, i.e.\ all subsets of $S$ except the empty set. 
In the $i$-th graph, every pair of nodes $X,Y$ 
such that $d(X,Y)\le i$ will be linked by an edge. 
Thus, an independent set of this graph will be composed of nodes representing 
fingerprints whose minimum separation is at least $i+1$. 
A maximal\footnote{The maximum independent set problem is NP-hard, 
so it is infeasible to exactly compute a maximum independent set of every graph. 
Alternatively, we use a well known greedy approximation, 
which consists in iteratively finding a minimum-degree non-isolated vertex, 
and removing all its neighbours, 
until obtaining an empty graph, whose vertex set is an independent 
set of the original graph.} independent set of the fingerprint graph 
is computed in every iteration, 
to have an approximation of a maximum-cardinality set of uniformly 
distributed fingerprints with minimum separation at least $i+1$. 
For example, in the graph of Figure~\ref{fig-fp-graph} (a), 
the method will find $\{\{1\},\{2\},\{3\},\{1,2,3\}\}$ as a maximum-cardinality 
set of uniformly distributed fingerprints with minimum separation $2$; 
whereas for the graph of Figure~\ref{fig-fp-graph} (b), 
the method will find, for instance, $\{\{1\},\{2,3\}\}$ 
as a maximum-cardinality set of uniformly distributed fingerprints 
with minimum separation $3$. 
The method iterates until finding the smallest maximal independent set 
that still contains at least $m$ fingerprints, that is, at least 
as many fingerprints as original victims, so a different fingerprint 
can be assigned to every victim.
If this set does not contain 
exactly $m$ fingerprints, it is used as a pool, from which successive 
runs of the attack randomly draw $m$ fingerprints. 
Algorithm~\ref{alg-unif-distr-fps} lists the pseudo-code of this method. 

\begin{algorithm}
\caption{Given a set $S$ of sybil nodes, compute a uniformly distributed set 
of fingerprints of size at least $m$.} 
\label{alg-unif-distr-fps}
\begin{algorithmic}[1]
\State $i\gets 1$
\State $G_{F}^{(1)}\gets (\powerset{S}\setminus \emptyset, 
\{(X,Y)\;\mid\; X,Y\in \powerset{S}\setminus \emptyset, X\neq Y, d(X,Y)\le 1\})$ 
\State $I_1\gets\textsc{MaxIndSet}\left(G_{F}^{(1)}\right)$
\Repeat\label{st-repeat}
  \State $F\gets I_i$  
  \State $i\gets i+1$
  \State $G_{F}^{(i)}\gets (\powerset{S}\setminus \emptyset, 
  \{(X,Y)\;\mid\; X,Y\in \powerset{S}\setminus \emptyset, X\neq Y, d(X,Y)\le i\})$ 
  \State $I_i\gets\textsc{MaxIndSet}\left(G_{F}^{(i)}\right)$
\Until{$|I_i|<m$}\label{st-until} 
\State \Return $F$
\Statex\hrulefill
\end{algorithmic}

\begin{algorithmic}[1]
\Function{MaxIndSet}{$G=(V_G,E_G)$}
\Repeat
  \State $v\gets\arg\min_{v\in V_G, \delta_G(v)\neq 0}\left\{\delta_G(v)\right\}$ 
  \State $E_G\gets E_G\setminus\{(v,w)\;|\;w\in N_G(v)\}$
  \State $V_G\gets V_G\setminus N_G(v)$
\Until{$E_G=\emptyset$}
\State \Return $V_G$
\EndFunction
\end{algorithmic}
\end{algorithm} 

In Algorithm~\ref{alg-unif-distr-fps}, the order of every graph $G_F^{(i)}$ 
is $2^{|S|}-1$. Thus, the time complexity of every graph construction 
is $\mathcal{O}\left(\left(2^{|S|}\right)^2\right)
=\mathcal{O}\left(2^{2|S|}\right)$. Moreover, the greedy algorithm for finding 
a maximal independent set runs in quadratic 
time with respect to the order of the graph, so in this case its time complexity 
is also $\mathcal{O}\left(2^{2|S|}\right)$. Finally, since the maximum possible 
distance between a pair of fingerprints is $|S|$, the worst-case 
time complexity of Algorithm~\ref{alg-unif-distr-fps} 
is $\mathcal{O}\left(|S|\cdot 2^{2|S|}\right)$. 
This worst case occurs when the number of victims is very small, as the number 
of times that steps~\ref{st-repeat} to~\ref{st-until} of the algorithm 
are repeated is more likely to approach $|S|$. 
While this time complexity may appear as excessive at first glance, 
we must consider that, for a social graph of order $n$, the algorithm 
will be run for sets of sybil nodes having at most cardinality $|S|=\log_2 n$. 
Thus, in terms of the order of the social graph, the worst-case running time 
will be $\mathcal{O}(n^2 \log_2 n)$. 

\subsection{Attacker subgraph retrieval}\label{sec-robust-retrieval} 

As discussed in Subsection~\ref{subsec-syb-retr-ideal}, 
in the original formulation of active attacks, the sybil retrieval phase 
is based on the assumption that the attacker subgraph can be uniquely 
and exactly matched to a subgraph of the released graph. 
This assumption is relaxed by the formulation of robust attacker subgraph 
retrieval given in Definition~\ref{def-ideal-sybil-subgraph-retrieval}, 
which accounts for the possibility that the attacker subgraph has been perturbed. 
The problem formulation given 
in Definition~\ref{def-ideal-sybil-subgraph-retrieval} requires 
a dissimilarity measure $\dist$ to compare candidate subgraphs 
to the original attacker subgraph. We will introduce such a measure 
in this section. Moreover, the problem formulation requires 
searching the entire power set of $V_{\transfg}$, which is infeasible 
in practice. In order to reduce the size of the search space, we will establish 
a perturbation threshold $\threshold$, and the search procedure will 
discard any candidate subgraph $X$ 
such that $\dist(\wis{X}{\transfg},\wis{S}{\plussybg})>\threshold$. 

We now define the dissimilarity measure $\dist$ that will be used. 
To that end, some additional notation will be necessary. 
For a graph $H$, a vertex set $V\subseteq V_H$, 
and a complete order $\prec\subseteq V\times V$, we will define 
the vector $\mathbf{v}_{\prec}=(v_{i_1}, v_{i_2}, \ldots, v_{i_{|V|}})$, 
as the one satisfying $v_{i_1} \prec v_{i_2} \prec \ldots \prec v_{i_{|V|}}$. 
When the order $\prec$ is fixed or clear from the context, we will simply 
refer to $\mathbf{v}_{\prec}$ as $\mathbf{v}$. 
Moreover, for the sake of simplicity in presentation, we will in some cases 
abuse notation and use $\mathbf{v}$ for $V$, 
$\wis{\mathbf{v}}{H}$ for $\wis{V}{H}$, and so on. 
The search procedure assumes the existence of a fixed order $\prec_S$ 
on the original set of sybil nodes $S$, 
which is established at the attacker subgraph creation stage, as discussed 
in Subsection~\ref{sec-ssgr-fp-creation}. 
In what follows, we will use the notation $\mathbf{s}=(x_1,x_2,\ldots,x_{|S|})$ 
for the vector~$\mathbf{s}_{\prec_S}$. 

Given the original attacker subgraph $\wis{S}{\plussybg}$ and a subgraph 
of $\transfg$ weakly induced by a candidate vector 
$\mathbf{v}=(v_1,v_2,\ldots,v_{|S|})$, the dissimilarity measure $\dist$ 
will compare $\wis{\mathbf{v}}{\transfg}$ to $\wis{S}{\plussybg}$ according 
to the following criteria: 
\begin{itemize}
\item The set of \emph{inter-sybil edges} of $\wis{S}{\plussybg}$ will be compared 
to that of $\wis{\mathbf{v}}{\transfg}$. This is equivalent to comparing 
$E(\indsg{S}{\plussybg})$ and $E(\indsg{\mathbf{v}}{\transfg})$. 
To that end, we will apply to $\indsg{S}{\plussybg}$ 
the isomorphism $\varphi'\colon\indsg{S}{\plussybg}\rightarrow 
\indsg{\mathbf{v}}{\transfg}$, which 
makes $\varphi'(x_i)=v_i$ for every $i\in\{1,\ldots,|S|\}$. 
The contribution of inter-sybil edges to $\dist$ will thus 
be defined as 
\begin{equation}\label{eq-wedist-inner}
\appredistsyb\left(\wis{S}{\plussybg},\wis{\mathbf{v}}{\transfg}\right)=
\left|
E(\varphi'\indsg{S}{\plussybg}) \symdif E(\indsg{\mathbf{v}}{\transfg})
\right|,
\end{equation} 
that is, the symmetric difference between the edge sets 
of $\varphi'\indsg{S}{\plussybg}$ and $\indsg{\mathbf{v}}{\transfg}$. 
\item The set of \emph{sybil-to-non-sybil edges} of $\wis{S}{\plussybg}$ 
will be compared to that of $\wis{\mathbf{v}}{\transfg}$. 
Unlike the previous case, where the orders $\prec_{S}$ and $\prec_{\mathbf{v}}$ 
allow to define a trivial isomorphism between the induced subgraphs, 
in this case creating the appropriate matching would be equivalent 
to solving the re-identification problem for every candidate subgraph, 
which is considerably inefficient. 
In consequence, we introduce a relaxed criterion, 
which is based on the numbers of non-sybil neighbours of every sybil node, 
which we refer to as \emph{marginal degrees}. The marginal degree 
of a sybil node $x\in S$ is thus defined 
as $\delta_{\wis{S}{\plussybg}}'(x)=\left|
N_{\wis{S}{\plussybg}}(x)\setminus S\right|$. By analogy, 
for a vertex $v\in\mathbf{v}$, 
we define $\delta_{\wis{\mathbf{v}}{\transfg}}'(v)=\left|
N_{\wis{\mathbf{v}}{\transfg}}(v)\setminus \mathbf{v}\right|$. 
Finally, the contribution of sybil-to-non-sybil edges to $\dist$ 
will be defined as 
\begin{equation}\label{eq-wedist-neigh}
\appredistneigh\left(\wis{S}{\plussybg},\wis{\mathbf{v}}{\transfg}\right)=
\displaystyle\sum_{i=1}^{|S|}\left|\delta_{\wis{\mathbf{v}}{\transfg}}'(v_i) - 
\delta_{\wis{S}{\plussybg}}'(x_i)\right| 
\end{equation}
\item The dissimilarity measure combines the previous criteria as follows: 
\begin{equation}\label{eq-wedist}
\appredist\left(\wis{S}{\plussybg},\wis{\mathbf{v}}{\transfg}\right)=
\appredistsyb\left(\wis{S}{\plussybg},\wis{\mathbf{v}}{\transfg}\right)+
\appredistneigh\left(\wis{S}{\plussybg},\wis{\mathbf{v}}{\transfg}\right)
\end{equation}
\end{itemize}

Figure~\ref{fig-ex-dissim-measure} shows an example of the computation 
of this dissimilarity measure, with $\mathbf{s}=(x_1,x_2,x_3,$ $x_4,x_5)$ 
and $\mathbf{v}=(v_1,v_2,v_3,v_4,v_5)$. In the figure, we can observe 
that $(x_1,x_3)\in E_{\wis{S}{\plussybg}}$ 
and $(v_1,v_3)\notin E_{\wis{\mathbf{v}}{\transfg}}$,  
whereas $(x_3,x_4)\in E_{\wis{S}{\plussybg}}$ 
and $(v_3,v_4)\notin E_{\wis{\mathbf{v}}{\transfg}}$. 
In consequence,\\ 
$\appredistsyb\left(\wis{S}{\plussybg},\wis{\mathbf{v}}{\transfg}\right)=2$. 
Moreover, we can also observe that $\delta_{\wis{S}{\plussybg}}'(x_2)=
|\emptyset|=0$, whereas $\delta_{\wis{\mathbf{v}}{\transfg}}'(v_2)=|
\{y'_1,y'_5\}|=2$. 
Since $\delta_{\wis{S}{\plussybg}}'(x_i)=
\delta_{\wis{\mathbf{v}}{\transfg}}'(v_i)$ for $i\in\{1,3,4,5\}$, 
we have 
$\appredistneigh\left(\wis{S}{\plussybg},\wis{\mathbf{v}}{\transfg}\right)=2$, 
so $\appredist\left(\wis{S}{\plussybg},\wis{\mathbf{v}}{\transfg}\right)=4$.  
It is simple to see that the value of the dissimilarity function 
is dependent on the order imposed by the vector $\mathbf{v}$. 
For example, consider the vector $\mathbf{v}'=(v_5,v_2,v_3,v_4,v_1)$. 
We can verify\footnote{We now have that 
$(x_1,x_2)\in E_{\wis{S}{\plussybg}}$ 
and $(v_5,v_2)\notin E_{\wis{\mathbf{v}'}{\transfg}}$; 
$(x_1,x_3)\in E_{\wis{S}{\plussybg}}$ 
and $(v_5,v_3)\notin E_{\wis{\mathbf{v}'}{\transfg}}$; 
$(x_2,x_5)\notin E_{\wis{S}{\plussybg}}$ 
and $(v_2,v_1)\in E_{\wis{\mathbf{v}'}{\transfg}}$; and 
$(x_3,x_4)\in E_{\wis{S}{\plussybg}}$ 
whereas $(v_3,v_4)\notin E_{\wis{\mathbf{v}'}{\transfg}}$. 
Moreover, now $\left|\delta_{\wis{\mathbf{v}'}{\transfg}}'(v_5) - 
\delta_{\wis{S}{\plussybg}}'(x_1)\right|=1$, 
$\left|\delta_{\wis{\mathbf{v}'}{\transfg}}'(v_2) - 
\delta_{\wis{S}{\plussybg}}'(x_2)\right|=2$ and  
$\left|\delta_{\wis{\mathbf{v}'}{\transfg}}'(v_1) - 
\delta_{\wis{S}{\plussybg}}'(x_5)\right|~=~1$.} 
that now 
$\appredistsyb\left(\wis{S}{\plussybg},\wis{\mathbf{v}'}{\transfg}\right)=4$, 
whereas 
$\appredistneigh\left(\wis{S}{\plussybg},\wis{\mathbf{v}'}{\transfg}\right)=4$, 
so the dissimilarity value is now 
$\appredist\left(\wis{S}{\plussybg},\wis{\mathbf{v}'}{\transfg}\right)=8$.

\begin{figure}[!ht]

\begin{center}

\begin{tikzpicture}[inner sep=0.7mm, place/.style={circle,draw=black,
fill=white},xx/.style={circle,draw=black!99, 
fill=black!99},gray1/.style={circle,draw=black!99, 
fill=black!25},gray2/.style={circle,draw=black!99, 
fill=black!50},gray3/.style={circle,draw=black!99,fill=black!75},
gray2d/.style={diamond,draw=black!99,fill=black!50},
transition/.style={rectangle,draw=black!50,fill=black!20,thick}, 
line width=.5pt]

\coordinate (x1) at (0,0);
\coordinate (x2) at (1,0);
\coordinate (x3) at (2,0);
\coordinate (x4) at (3,0);
\coordinate (x5) at (4,0);

\coordinate (y11) at (-0.8,0.8);
\coordinate (y12) at (1,1);
\coordinate (y13) at (3,1);
\coordinate (y14) at (2,-1);

\coordinate (v1) at (6,0);
\coordinate (v2) at (7,0);
\coordinate (v3) at (8,0);
\coordinate (v4) at (9,0);
\coordinate (v5) at (10,0);

\coordinate (y21) at (5.2,0.8);
\coordinate (y22) at (7,1);
\coordinate (y23) at (9,1);
\coordinate (y24) at (8,-1);
\coordinate (y25) at (7,-1);

\draw[black] (x1) -- (x5);
\draw[black] (x1) .. controls (1,-1) .. (x3);
\draw[black] (x1) .. controls (1.5,0.64) .. (x4);
\draw[black] (x2) .. controls (2,1) .. (x4);
\draw[black] (y11) -- (x1) -- (y12) -- (x3) -- (y13) -- (x5);
\draw[black] (x3) -- (y14);

\draw[black] (v1) -- (v3);
\draw[black] (v4) -- (v5);
\draw[black] (v1) .. controls (7.5,0.64) .. (v4);
\draw[black] (v2) .. controls (8,1) .. (v4);
\draw[black] (y21) -- (v1) -- (y22) -- (v3) -- (y23) -- (v5);
\draw[black] (v3) -- (y24);
\draw[black] (v2) -- (y21);
\draw[black] (v2) -- (y25);

\foreach \ind in {1,...,5} \node [gray1] at (x\ind) {};
\foreach \ind in {1,...,4} \node [place] at (y1\ind) {};
\foreach \ind in {1,...,5} \node [gray1] at (v\ind) {};
\foreach \ind in {1,...,5} \node [place] at (y2\ind) {};

\node at (-0.3,-0.3) {$x_1$};
\node at (1.3,-0.3) {$x_2$};
\node at (2.3,-0.3) {$x_3$};
\node at (3.3,-0.3) {$x_4$};
\node at (4.3,-0.3) {$x_5$};

\node at (-1,1) {$y_1$};
\node at (1,1.3) {$y_2$};
\node at (3,1.3) {$y_3$};
\node at (2,-1.3) {$y_4$};

\node at (6.3,-0.3) {$v_1$};
\node at (7.3,-0.3) {$v_2$};
\node at (8.3,-0.3) {$v_3$};
\node at (9.3,-0.3) {$v_4$};
\node at (10.3,-0.3) {$v_5$};

\node at (5,1) {$z_1$};
\node at (7,1.3) {$z_2$};
\node at (9,1.3) {$z_3$};
\node at (8,-1.3) {$z_4$};
\node at (7,-1.3) {$z_5$};

\coordinate [label=center:{\small{$\wis{S}{\plussybg}$}}] (al) at (2,-2);
\coordinate [label=center:{\small{$\wis{\mathbf{v}}{\transfg}$}}] (bl) at (8,-2);

\end{tikzpicture}

\end{center}

\caption{An example of possible graphs $\wis{S}{\plussybg}$ 
and $\wis{\mathbf{v}}{\transfg}$. 
Vertices in $S$ and $\mathbf{v}$ are coloured in gray.} 
\label{fig-ex-dissim-measure} 
\end{figure} 

The search procedure assumes 
that the transformation $t$ did not remove the image of any sybil node 
from $\iso\plussybg$, so it searches the set 
of cardinality-$|S|$ permutations of elements from~$V_{\transfg}$,  
respecting the tolerance threshold. The method is a breadth-first search, 
which analyses at the $i$-th level the possible matches to the vector 
$(x_1,x_2,\ldots,x_i)$ composed of the first $i$ components of $\mathbf{s}$. 
The tolerance threshold $\threshold$ is used to prune the search tree. 
A detailed description of the procedure is shown 
in Algorithm~\ref{alg-bdf-search}. 
Ideally, the algorithm outputs a unitary set 
$\tilde{\mathcal{C}}^*=\{(v_{j_1},v_{j_2},\ldots,v_{j_{|S|}})\}$, 
in which case the vector $\mathbf{v}=(v_{j_1},v_{j_2},\ldots,v_{j_{|S|}})$ 
is used as the input to the fingerprint matching phase, described in the following 
subsection. If this is not the case, and the algorithm yields  
$\tilde{\mathcal{C}}^*=\{\mathbf{v}_1,\mathbf{v}_2,\ldots,\mathbf{v}_t\}$, 
the attack randomly picks an element $\mathbf{v}_i\in\tilde{\mathcal{C}}^*$ 
and proceeds to the fingerprint matching phase. Finally, 
if $\tilde{\mathcal{C}}^*=\emptyset$, the attack is considered to fail, 
as no re-identification is possible. To conclude, we point out that 
if Algorithm~\ref{alg-bdf-search} is run 
with $\threshold=0$, then $\tilde{\mathcal{C}}^*$ contains 
exactly the same candidate set that would be recovered 
by the attacker subgraph retrieval phase of the original walk-based attack. 

\begin{algorithm} 
\caption{Given the graphs $\plussybg$ and $\transfg$, the set of original 
sybil nodes $S\subseteq V_{\plussybg}$, and the maximum distance threshold 
$\threshold$, obtain the set $\tilde{\mathcal{C}}^*$ 
of equally-likely best candidate sybil sets.} 
\label{alg-bdf-search} 
\begin{algorithmic}[1] 
\State $\triangleright$ Find suitable candidates to match $x_1$ 
\State $PartialCandidates_1\gets\emptyset$
\State $d\gets\threshold$
\For{$v\in V_{\transfg}$}
  \If{$\appredist\left(\wis{(x_1)}{\plussybg},\wis{(v)}{\transfg}\right)<d$}
    \State $PartialCandidates_1\gets\{(v)\}$
    \State $d\gets 
    \appredist\left(\wis{(x_1)}{\plussybg},\wis{(v)}{\transfg}\right)$
  \ElsIf{$\appredist\left(\wis{(x_1)}{\plussybg},\wis{(v)}{\transfg}\right)=d$}
    \State $PartialCandidates_1\gets PartialCandidates_1\cup\{(v)\}$
  \EndIf
\EndFor
\If{$PartialCandidates_1=\emptyset$}
  \State \Return $\emptyset$
\ElsIf{$|S|=1$}
  \State \Return $PartialCandidates_1$
\Else
  \State $\triangleright$ Find rest of matches for candidates
  \State \Return \Call{Breadth-First-Search}{$2$, $PartialCandidates_1$}
  \ \ \ \ \ \ $\triangleright$ Algorithm~\ref{alg-bdf-search-fn-bfs}
\EndIf

\end{algorithmic}

\end{algorithm} 

\begin{algorithm}
\caption{Function \textsc{Breadth-First-Search}, 
used in Algorithm~\ref{alg-bdf-search}.} 
\label{alg-bdf-search-fn-bfs} 

\begin{algorithmic}[1]
\Function{Breadth-First-Search}{$i$, $PartialCandidates_{i-1}$}
  \State $\triangleright$ Find suitable candidates to match $(x_1,x_2,\ldots,x_i)$  
  \State $\mathbf{s}'\gets(x_1,x_2,\ldots,x_i)$
  \State $PartialCandidates_i\gets\emptyset$
  \State $d\gets\threshold$ 
  \For{$(v_{j_1},v_{j_2},\ldots,v_{j_{i-1}})\in \mathcal{C}_{i-1}$} 
    \State $ExtendedCandidates\gets\emptyset$
    \State $d'\gets\threshold$
    \For{$w\in V_{\transfg}\setminus(v_{j_1},v_{j_2},\ldots,v_{j_{i-1}})$}
      \State $\mathbf{v}'\gets (v_{j_1},v_{j_2},\ldots,v_{j_{i-1}},w)$
      \If{$\appredist\left(\wis{\mathbf{s}'}{\plussybg},
      \wis{\mathbf{v}'}{\transfg}\right)<d'$}
        \State $ExtendedCandidates\gets\{\mathbf{v}'\}$
        \State $d'\gets \appredist\left(\wis{\mathbf{s}'}{\plussybg},
        \wis{\mathbf{v}'}{\transfg}\right)$
      \ElsIf{$\appredist\left(\wis{\mathbf{s}'}{\plussybg},
      \wis{\mathbf{v}'}{\transfg}\right)=d'$}
        \State $ExtendedCandidates\gets ExtendedCandidates\cup\{\mathbf{v}'\}$
      \EndIf
    \EndFor
    \If{$d'<d$}
      \State $PartialCandidates_i\gets ExtendedCandidates$
      \State $d\gets d'$
    \ElsIf{$d'=d$}  
      \State $PartialCandidates_i\gets PartialCandidates_i\cup ExtendedCandidates$
    \EndIf
  \EndFor
  \If{$PartialCandidates_i=\emptyset$}
    \State \Return $\emptyset$
  \ElsIf{$i=|S|$}
    \State \Return $PartialCandidates_i$
  \Else
    \State $\triangleright$ Find rest of matches for candidates
    \State \Return \Call{Breadth-First-Search}{$i+1$, $PartialCandidates_i$}\
  \EndIf
\EndFunction
\end{algorithmic} 

\end{algorithm}

\subsection{Fingerprint matching}\label{sec-robust-re-identification} 

Now, we describe the noise-tolerant fingerprint matching process. 
Let $Y=\{y_1,\ldots,y_m\}\subseteq V_{\plussybg}$ represent the set of 
victims. 
Let $S$ be the original set of sybil nodes and $\tilde{S}'\subseteq V_{\transfg}$ 
a candidate obtained by the sybil retrieval procedure described above. 
As in the previous subsection, let $\mathbf{s}=(x_1,x_2,\ldots,x_{|S|})$ 
be the vector containing the elements of $S$ in the order imposed at the sybil 
subgraph creation stage. Moreover, let $\wis{\mathbf{v}}{\transfg}$, 
with $\mathbf{v}=(v_1,v_2,\ldots,v_{|S|})\in\tilde{\mathcal{C}}^*$, 
be a candidate sybil subgraph, retrieved using Algorithm~\ref{alg-bdf-search}.  
Finally, for every $i\in\{1,\ldots,m\}$, let $F_i\subseteq S$ be the original 
fingerprint of the victim $y_i$ and $\phi F_i\subseteq \mathbf{v}$ its image 
by the isomorphism mapping $\mathbf{s}$ to $\mathbf{v}$. 

We now describe the process for finding 
$Y'=\{y'_1,\ldots,y'_m\}\subseteq V_{\transfg}$, 
where $y'_i=\iso(y_i)$, using $\phi F_1, \phi F_2, \ldots, \phi F_m$, 
$\mathbf{s}$ and $\mathbf{v}$. 
If the perturbation $\transf{\plussybg}$ had caused no damage to the fingerprints, 
checking for the exact matches is sufficient. Since, as previously discussed, 
this is usually not the case, we will introduce a noise-tolerant 
fingerprint matching strategy that maps every original 
fingerprint to its most similar candidate fingerprint, within some tolerance 
threshold $\thresholdfp$. 

Algorithm~\ref{alg-fp-matching} describes the process for finding 
the set of optimal re-identifications. 
For a candidate victim $z\in N_{\transfg}(\mathbf{v})\setminus\mathbf{v}$, 
the algorithm denotes as $\tilde{F}_{z,\mathbf{v}}=N_{\transfg}(z)\cap \mathbf{v}$ 
its fingerprint with respect to $\mathbf{v}$. The algorithm is a depth-first 
search procedure. First, the algorithm finds, 
for every $\phi F_i$, $i\in\{1,\ldots,m\}$, 
the set of most similar candidate fingerprints, and keeps the set of matches 
that reach the minimum distance. From these best matches, one or several 
partial re-identifications are obtained. 
The reason why more than one partial re-identification is obtained 
is that more than one candidate fingerprint may be equally similar 
to some $\phi F_i$. For every partial re-identification, 
the method recursively finds the set of best completions 
and combines them to construct the final set of equally likely 
re-identifications. The search space is reduced 
by discarding insufficiently similar matches. For any candidate victim $z$ 
and any original victim $y_i$ such that 
$d(\tilde{F}_{z,\mathbf{v}},\phi F_i)<\thresholdfp$, 
the algorithm discards all matchings where $\iso(y_i)=z$. 

To illustrate how the method works, recall the graphs $\wis{S}{\plussybg}$ 
and $\wis{\mathbf{v}}{\transfg}$ depicted in Figure~\ref{fig-ex-dissim-measure}. 
The original set of victims is $Y=\{y_1,y_2,y_3,y_4\}$ and their fingerprints 
are $F_1=\{x_1\}$, $F_2=\{x_1,x_3\}$, $F_3=\{x_3,x_5\}$, and $F_4=\{x_3\}$, 
respectively. In consequence, we have $\phi F_1=\{v_1\}$, $\phi F_2=\{v_1,v_3\}$, 
$\phi F_3=\{v_3,v_5\}$, and $\phi F_4=\{v_3\}$. The set of candidate victims 
is $N_{\transfg}(\mathbf{v})\setminus\mathbf{v}=\{z_1,z_2,z_3,z_4,z_5\}$. 
The method will first find all exact matchings, that is $\varphi(y_2)=z_2$, 
$\varphi(y_3)=z_3$, and $\varphi(y_4)=z_4$, because the distances between 
the corresponding fingerprints is zero in all three cases. 
Since none of these matchings 
is ambiguous, the method next determines the match $\varphi(y_1)=z_1$, 
because $d(\tilde{F}_{z_1},\phi F_1)=d(\{v_1,v_2\},\{v_1\})=1
<2=d(\{v_2\},\{v_1\})=d(\tilde{F}_{z_5},\phi F_1)$. 
At this point, the method stops and yields the unique re-identification 
$\{(y_1,z_1),(y_2,z_2),(y_3,z_3),(y_4,z_4)\}$. 
Now, suppose that the vertex $z_5$ is linked in $\transfg$ to $v_3$, 
instead of $v_2$, as depicted in Figure~\ref{fig-ex-ambigous-re-idef}. 
In this case, the method will unambiguously determine the matchings 
$\varphi(y_2)=z_2$ and $\varphi(y_3)=z_3$, and then will try the two 
choices $\varphi(y_4)=z_4$ and $\varphi(y_4)=z_5$. In the first case, 
the method will make $\varphi(y_1)=z_1$ and discard $z_5$. 
Analogously, in the second case the method will also make $\varphi(y_1)=z_1$,  
and will discard $z_4$. 
Thus, the final result will consist in two equally likely re-identifications, 
namely $\{(y_1,z_1),(y_2,z_2),(y_3,z_3),(y_4,z_4)\}$ and 
$\{(y_1,z_1),(y_2,z_2),(y_3,z_3),(y_4,z_5)\}$.  

\begin{figure}[!ht]

\begin{center}

\begin{tikzpicture}[inner sep=0.7mm, place/.style={circle,draw=black,
fill=white},xx/.style={circle,draw=black!99, 
fill=black!99},gray1/.style={circle,draw=black!99, 
fill=black!25},gray2/.style={circle,draw=black!99, 
fill=black!50},gray3/.style={circle,draw=black!99,fill=black!75},
gray2d/.style={diamond,draw=black!99,fill=black!50},
transition/.style={rectangle,draw=black!50,fill=black!20,thick}, 
line width=.5pt]

\coordinate (x1) at (0,0);
\coordinate (x2) at (1,0);
\coordinate (x3) at (2,0);
\coordinate (x4) at (3,0);
\coordinate (x5) at (4,0);

\coordinate (y11) at (-0.8,0.8);
\coordinate (y12) at (1,1);
\coordinate (y13) at (3,1);
\coordinate (y14) at (2,-1);

\coordinate (v1) at (6,0);
\coordinate (v2) at (7,0);
\coordinate (v3) at (8,0);
\coordinate (v4) at (9,0);
\coordinate (v5) at (10,0);

\coordinate (y21) at (5.2,0.8);
\coordinate (y22) at (7,1);
\coordinate (y23) at (9,1);
\coordinate (y24) at (8.5,-1);
\coordinate (y25) at (7.5,-1);

\draw[black] (x1) -- (x5);
\draw[black] (x1) .. controls (1,-1) .. (x3);
\draw[black] (x1) .. controls (1.5,0.64) .. (x4);
\draw[black] (x2) .. controls (2,1) .. (x4);
\draw[black] (y11) -- (x1) -- (y12) -- (x3) -- (y13) -- (x5);
\draw[black] (x3) -- (y14);

\draw[black] (v1) -- (v3);
\draw[black] (v4) -- (v5);
\draw[black] (v1) .. controls (7.5,0.64) .. (v4);
\draw[black] (v2) .. controls (8,1) .. (v4);
\draw[black] (y21) -- (v1) -- (y22) -- (v3) -- (y23) -- (v5);
\draw[black] (v3) -- (y24);
\draw[black] (v2) -- (y21);
\draw[black] (v3) -- (y25);

\foreach \ind in {1,...,5} \node [gray1] at (x\ind) {};
\foreach \ind in {1,...,4} \node [place] at (y1\ind) {};
\foreach \ind in {1,...,5} \node [gray1] at (v\ind) {};
\foreach \ind in {1,...,5} \node [place] at (y2\ind) {};

\node at (-0.3,-0.3) {$x_1$};
\node at (1.3,-0.3) {$x_2$};
\node at (2.3,-0.3) {$x_3$};
\node at (3.3,-0.3) {$x_4$};
\node at (4.3,-0.3) {$x_5$};

\node at (-1,1) {$y_1$};
\node at (1,1.3) {$y_2$};
\node at (3,1.3) {$y_3$};
\node at (2,-1.3) {$y_4$};

\node at (6.3,-0.3) {$v_1$};
\node at (7.3,-0.3) {$v_2$};
\node at (8.4,-0.1) {$v_3$};
\node at (9.3,-0.3) {$v_4$};
\node at (10.3,-0.3) {$v_5$};

\node at (5,1) {$z_1$};
\node at (7,1.3) {$z_2$};
\node at (9,1.3) {$z_3$};
\node at (8.5,-1.3) {$z_4$};
\node at (7.5,-1.3) {$z_5$};

\coordinate [label=center:{\small{$\wis{S}{\plussybg}$}}] (al) at (2,-2);
\coordinate [label=center:{\small{$\wis{\mathbf{v}}{\transfg}$}}] (bl) at (8,-2);

\end{tikzpicture}

\end{center}

\caption{Alternative example of possible graphs $\wis{S}{\plussybg}$ 
and $\wis{\mathbf{v}}{\transfg}$.} 
\label{fig-ex-ambigous-re-idef} 
\end{figure} 

\begin{algorithm} 
\caption{Given the graphs $\plussybg$ and $\transfg$, the original 
set of victims $Y=\{y_1,y_2,\ldots,y_m\}$, the original fingerprints 
$F_1, F_2, \ldots, F_m$, a candidate set of sybils $\mathbf{v}$, 
and the maximum distance threshold 
$\thresholdfp$, obtain the set $ReIdents$ of best matchings.} 
\label{alg-fp-matching} 
\begin{algorithmic}[1] 
\Statex
\State $(ReIdents, d)\gets $ 
\Call{GreedyMatching}{$Y$, $N_{\transfg}(\mathbf{v})\setminus\mathbf{v}$}
\State \Return $ReIdents$
\Statex\hrulefill
\end{algorithmic}

\begin{algorithmic}[1]
\Function{GreedyMatching}{$Y$, $M$}
\State $\triangleright$ Find best matches of some unmapped victim(s) 
to one or more candidate victims
\State $MapPartialBest\gets\emptyset$
\State $d\gets \beta$
\For{$y_i\in Y$}
\For{$z\in M$}
  \If{$d(\tilde{F}_{z,\mathbf{v}},\phi F_i)<d$}
    \State $MapPartialBest\gets\{(y_i,\{z\})\}$
    \State $d\gets d(\tilde{F}_{z,\mathbf{v}},\phi F_i)$
  \ElsIf{$d(\tilde{F}_{z,\mathbf{v}},\phi F_i)=d$}
    \If{$(y_i,P)\in MapPartialBest$}
      \State $MapPartialBest\gets(MapPartialBest\setminus 
      (y_i,P))\cup\{(y_i,P\cup\{z\})\}$ 
    \Else
      \State $MapPartialBest\gets MapPartialBest\cup\{(y_i,P\cup\{z\})\}$
    \EndIf
  \EndIf
\EndFor
\EndFor
\If{$MapPartialBest=\emptyset$}
  \State \Return $(\emptyset, +\infty)$
\Else
  \State $\triangleright$ Build partial re-identifications from best matches 
  \State Pick any $(y,P)\in MapPartialBest$
  \State $PartialReIdents\gets \{(y,z)\;|\;z\in P\}$
  \For{$(y',P')\in MapPartialBest\setminus (y,P)$}
    \State $PartialReIdents\gets\{R\cup \{y',z'\}
    \;|\;R\in PartialReIdents\wedge z'\in P'\}$
  \EndFor
  \If{$|Y|=1$}
    \State \Return $(PartialReIdents,d)$
  \Else
    \State $\triangleright$ Recursively find best completions for partial re-identifications  
    \State $BestComplReIdent \gets\emptyset$
    \State $d_{best}\gets \beta$
    \For{$R\in PartialReIdents$}
      \State $(CompletedReIdents, d)\gets $ 
      \Call{GreedyMatching}{$Y\setminus \{y\;|\;(y,z)\in R\}$, 
      $M\setminus \{z\;|\;(y,z)\in R\}$}
      \If{$d<d_{best}$}
        \State $BestComplReIdent\gets CompletedReIdents$
        \State $d_{best}\gets d$
      \ElsIf{$d=d_{best}$}
        \State $BestComplReIdent\gets BestComplReIdent\cup CompletedReIdents$
      \EndIf
    \EndFor
    \State \Return $(\{P\cup R\;|\; P\in PartialReIdents\wedge 
    R\in BestComplReIdent\}, d_{best})$
  \EndIf  
\EndIf
\EndFunction
\end{algorithmic}
\end{algorithm} 

Ideally, Algorithms~\ref{alg-bdf-search} 
and~\ref{alg-fp-matching} both yield unique solutions, in which case 
the sole element in the output of Algorithm~\ref{alg-fp-matching} 
is given as the final re-identification. If this is not the case, 
the attack picks a random candidate sybil subgraph from $\tilde{\mathcal{C}}^*$, 
uses it as the input of Algorithm~\ref{alg-fp-matching}, and picks a random 
re-identification from its output. 
If either algorithm yields an empty solution, the attack fails.  
Finally, it is important to note that, if Algorithm~\ref{alg-bdf-search} 
is run with $\threshold=0$ and Algorithm~\ref{alg-fp-matching} is run 
with $\thresholdfp=0$, then the final result is exactly the same 
set of equally likely matchings that would be obtained 
by the original walk-based attack. 

\section{Experiments}\label{sec-experiments} 

The purpose of our experiments is to show the considerable gain in resiliency 
of the proposed robust active attack, in comparison to the original attack. 
We used for our experiments a collection of randomly generated synthetic graphs. 
This collection is composed of $10,000$ graphs for each density value 
in the set $\{0.05, 0.1, \ldots, 1\}$. Each graph has $200$ 
vertices\footnote{We chose to use graphs of order $200$ to make the results 
comparable to those reported in~\cite{MRT2018} for the three anonymisation 
methods used here.}, and its edge set is randomly generated in such a manner 
that the desired density value is met. 

As discussed in~\cite{BDK2007}, for a graph having $n$ vertices, it suffices 
to insert $\log n$ sybil nodes for being able to compromise any possible 
victim, while it has been shown in~\cite{Yu2006,Yu2008} that even the 
so-called \emph{near-optimal} sybil defences do not aim to remove every 
sybil node, but to limit their number to around $\log n$.  
In light of these two considerations, when evaluating each attack strategy 
on the collection of randomly generated graphs, 
we use $8$ sybil nodes. Given the set $S$ of sybil nodes, the original attack 
randomly creates any set of fingerprints 
from $\powerset{S}\setminus \emptyset$. 
In the case of the robust attack, 
the pool of uniformly distributed fingerprints 
was generated in advance and, if it is larger than the number of victims, 
different sets of fingerprints of size $|Y|$ are randomly drawn 
from the pool at every particular run. Moreover, in Algorithms~\ref{alg-bdf-search} 
and~\ref{alg-fp-matching}, we made $\threshold=\thresholdfp=8$. 

For every graph in the collection, after simulating the attacker subgraph creation 
stage of each attack, and the pseudonymisation performed by the defender, 
we generate six variants of perturbed graphs using the following methods: 
\begin{enumerate}[(a)] 
\item The method proposed in~\cite{MauwTrujilloXuan2016} 
for enforcing $(k,\ell)$-anonymity for some $k>1$ or some $\ell>1$. 
\item The method proposed in~\cite{MRT2018} for enforcing 
$(2,\Gamma_{G,1})$-anonymity. 
\item The method proposed in~\cite{MRT2018} for enforcing 
$(k,\Gamma_{G,1})$-adjacency anonymity for a given value of $k$. 
Here, we will run the method with $k=|S|$, since it was empirically shown  
in~\cite{MRT2018} that the original walk-based attack is very likely 
to be thwarted in this case. 
\item Randomly flipping $1\%$ of the edges in $\plussybg$. Each flip 
consists in randomly selecting a pair of vertices $u,v\in V_{\plussybg}$, 
removing the edge $(u,v)$ if it belongs to $E_{\plussybg}$, or adding it 
in the opposite case. Since $\plussybg$ has order $n=208$, 
this perturbation 
performs $\left\lfloor 0.01\cdot \frac{n(n-1)}{2}\right\rfloor=215$ flips. 
\item Randomly flipping $5\%$ of the edges in $\plussybg$ (that is, $1076$ flips), 
in a manner analogous to the one used above. 
\item Randomly flipping $10\%$ of the edges in $\plussybg$ (that is, $2153$ flips), 
in a manner analogous to the one used above. 
\end{enumerate} 

Finally, we compute the probability of success of the re-identification 
stage for each perturbed variant. The success probability is computed as 

\begin{equation}\label{eq-succ-prob-robust}
\textstyle\Pr=\left\{\begin{array}{ll}
\frac{\sum_{X\in \mathcal{X}} 
p_{X}}{|\mathcal{X}|}
&\quad\text{if }\mathcal{X}\neq\emptyset\\
0&\quad\text{otherwise}
\end{array}
\right. 
\end{equation}

\noindent
where $\mathcal{X}$ is the set of equally-likely possible sybil subgraphs 
retrieved in $\transfg$ by the third phase of the attack\footnote{Note that, 
for the original walk-based attack, 
$\mathcal{X}$ is the set of subgraphs of $\transfg$ isomorphic 
to $\wis{S}{\plussybg}$, whereas for the robust attack we have  
$\mathcal{X}=\left\{\wis{\mathbf{v}}{\transfg}\;|\;
\mathbf{v}\in\tilde{\mathcal{C}}^*\right\}$, being $\tilde{\mathcal{C}}^*$ 
the output of Algorithm~\ref{alg-bdf-search}.}, and 

\begin{displaymath}
p_{X} =\left\{
\begin{array}{ll}
\frac{1}{|\mathcal{Y}_{X}|} \quad & \quad \text{if} 
\quad Y \in \mathcal{Y}_{X}\\
                  0 \quad & \quad \text{otherwise}
\end{array}
\right.
\end{displaymath}

\noindent
with $\mathcal{Y}_{X}$ containing all equally-likely 
matchings\footnote{Note that, for the original walk-based attack, 
$\mathcal{Y}_X=\{\{y_1,y_2\ldots,y_m\}\subseteq V_{\transfg}\;|\; 
\forall_{i\in\{1,\ldots,m\}} F_{y_i,X}=F_i\}$, whereas for the robust attack 
$\mathcal{Y}_X$ is the output of Algorithm~\ref{alg-fp-matching}.} 
according to $X$. 
These experiments\footnote{The implementations of the graph generators, 
anonymisation methods and attack simulations are available 
at \url{https://github.com/rolandotr/graph}.} 
were performed on the HPC platform 
of the University of Luxembourg~\cite{VBCG_HPCS14}. 

Figure~\ref{fig-succ-prob-rand-graphs-200-verts-8-syb} shows 
the success probabilities of the two attacks on the perturbed graphs obtained 
by applying the strategies (a) to (f). From the analysis of the figure, 
it is clear that 
the robust attack has a consistently larger probability of success 
than the original walk-based attack. 
A relevant fact evidenced by the figure (items (a), (b) and (c)) 
is that the robust attack 
is completely effective against the anonymisation methods described 
in~\cite{MauwTrujilloXuan2016,MRT2018}. As pointed out by the authors 
of these papers, the fact that the original walk-based attack leveraging 
more than one sybil node was thwarted to a considerable extent 
by the anonymisation methods was a side effect of the disruptions caused 
in the graph, rather than the enforced privacy measure itself. 
By successfully addressing this shortcoming, the robust active attack 
lends itself as a more appropriate benchmark on which to evaluate 
future anonymisation methods. For example, by comparing the charts 
in Figure~\ref{fig-succ-prob-rand-graphs-200-verts-8-syb} (a) and (b),  
we can see that the success probability of the original attack,
for some low density values,
is slightly higher for the second anonymisation method than for the first. 
Neither of these algorithms gives a theoretical privacy guarantee against 
an attacker leveraging seven sybil nodes. However, a difference 
in success probability is observed for the original attack, 
which is due to the fact that the second method introduces a smaller number 
of changes in the graph than the first one \cite{MRT2018}. 
As can be observed in the figure, the robust attack is not 
affected by this difference between the methods, and allows the analyst 
to reach the correct conclusion that both methods fail to thwart the attack.  
As a final observation, by analysing items 
(d) and (e) in Figure~\ref{fig-succ-prob-rand-graphs-200-verts-8-syb}, 
we can see that even $1\%$ of random noise completely thwarts the original 
attack, whereas the robust attack still performs at around $0.6$ success 
probability. The robust attack also performs acceptably well 
with a $5\%$ random perturbation. 

\begin{figure}[!ht]
\centering
\subfigure[t][$(k,\ell)$-anonymity, $k>1$ or $\ell>1$ \cite{MauwTrujilloXuan2016}]{
\includegraphics[scale=.57]{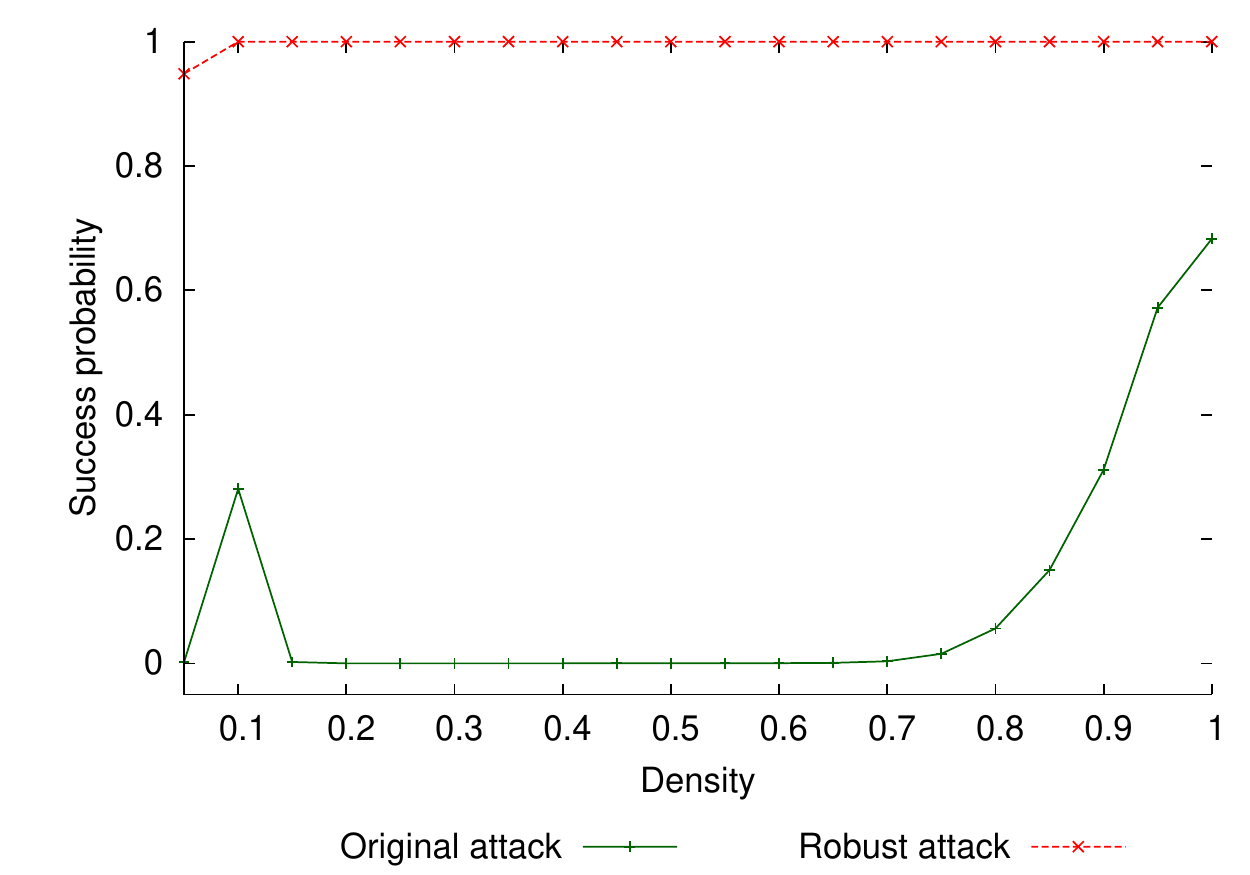}
}\hspace{-0.2cm}
\subfigure[t][$(2,\Gamma_1)$-anonymity \cite{MRT2018}]{
\includegraphics[scale=.57]{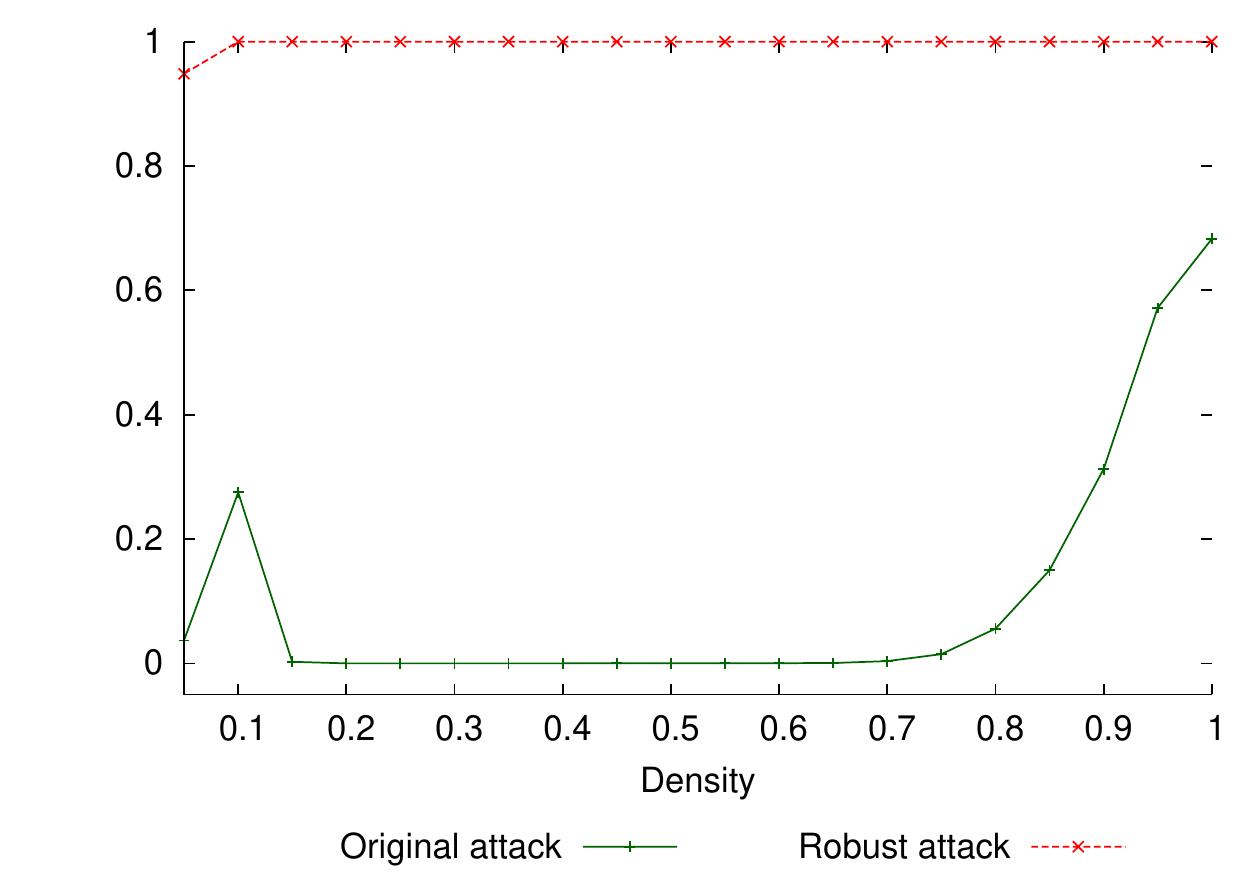}
}
\subfigure[t][$(|S|,\Gamma_1)$-adjacency anonymity \cite{MRT2018}]{
\includegraphics[scale=.57]{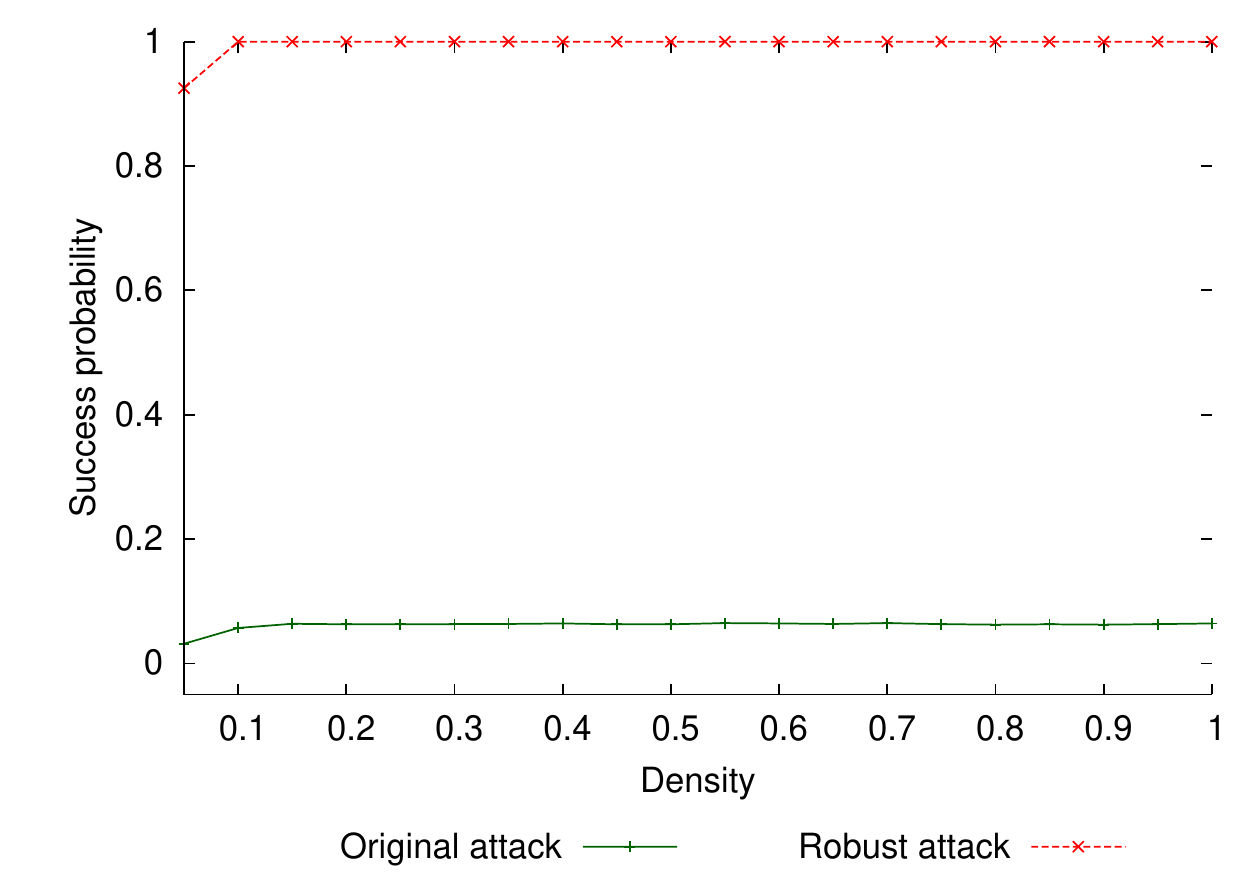}
}\hspace{-0.2cm}
\subfigure[t][$1\%$ random perturbation]{
\includegraphics[scale=.57]{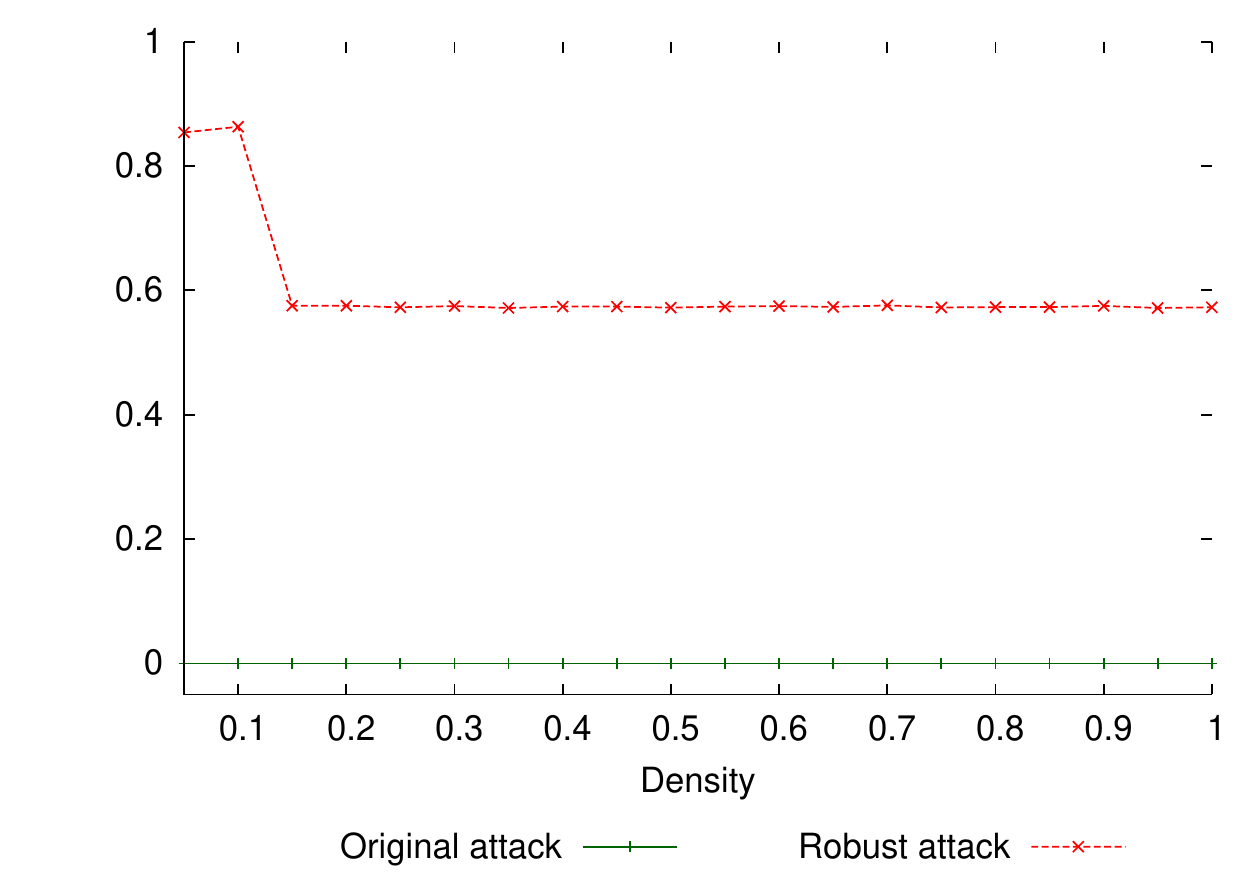}
}
\subfigure[t][$5\%$ random perturbation]{
\includegraphics[scale=.57]{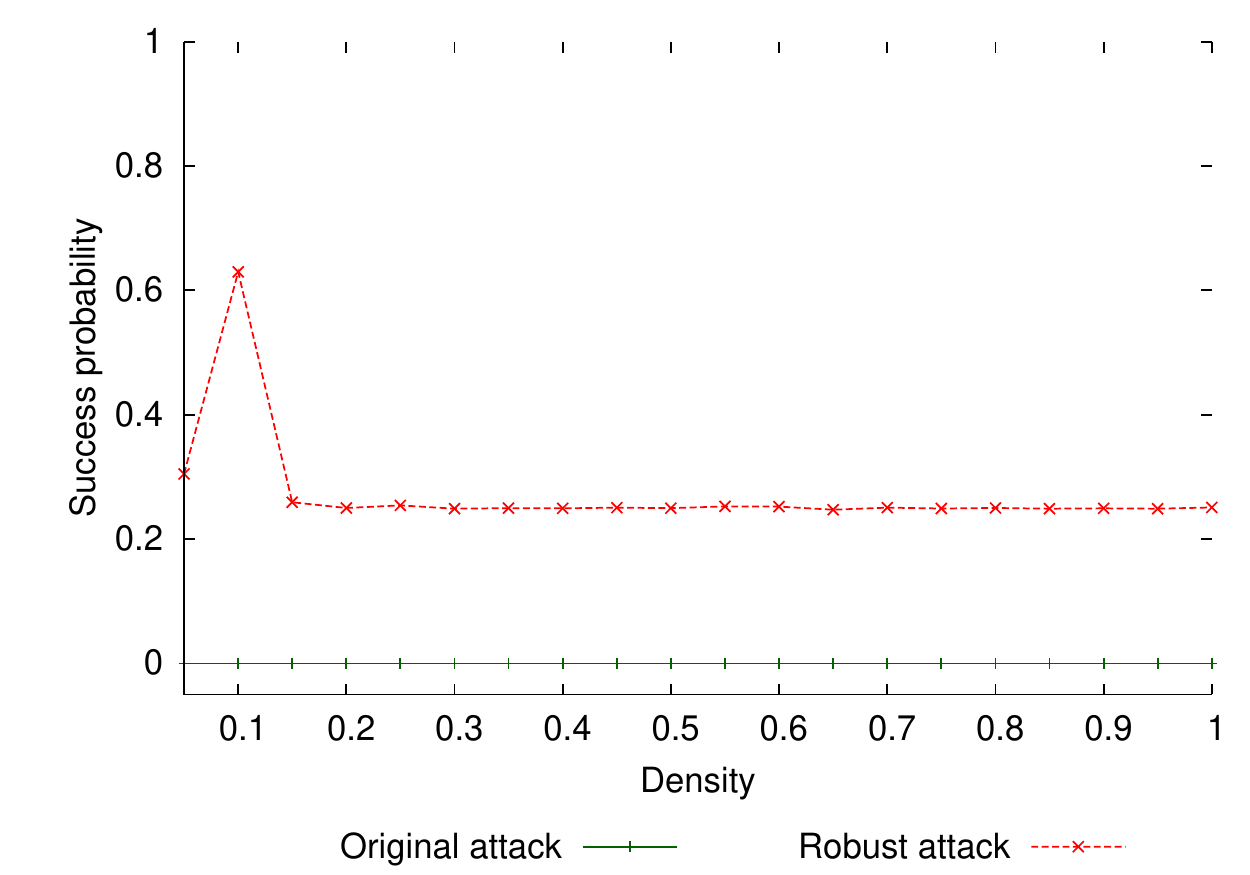}
}\hspace{-0.2cm}
\subfigure[t][$10\%$ random perturbation]{
\includegraphics[scale=.57]{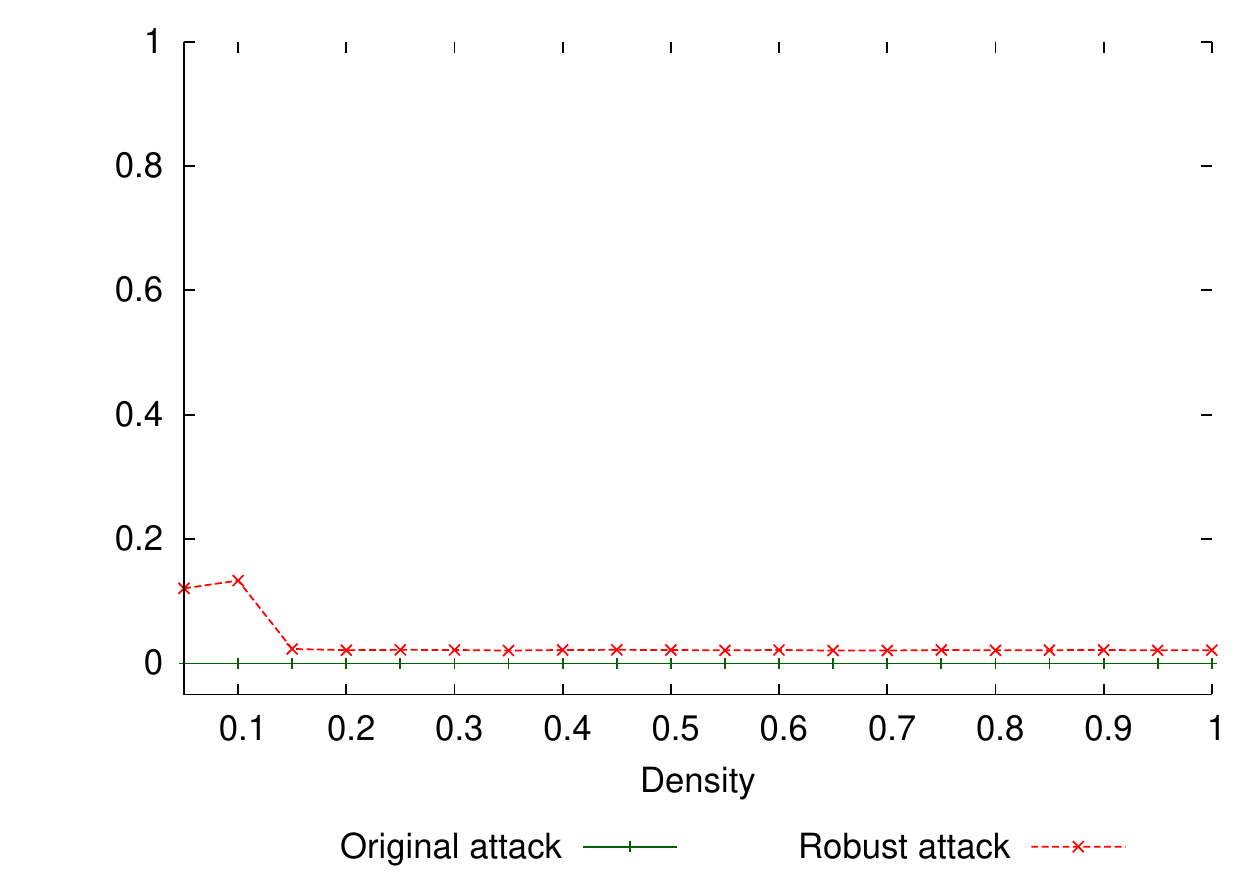}
}
\caption{Success probabilities of every attack variant on the collection 
of randomly generated graphs, after publishing the graphs perturbed 
by the methods listed above.} 
\label{fig-succ-prob-rand-graphs-200-verts-8-syb} 
\end{figure} 

\section{Related work} 
\label{sec-related-work} 

Privacy attacks on social networks exploit structural knowledge about the victims 
for re-identifying them in a released version of the social graph. These attacks 
can be divided in two categories, according to the manner in which the adversary 
obtains the knowledge used to re-identify the victims. On the one hand, 
\emph{passive attacks} rely on existing knowledge, which can be collected 
from publicly available sources, such as the public view of another social 
network where the victims are known to have accounts. The use of this type 
of information was demonstrated in~\cite{NS2009}, 
where information from Flickr was used to re-identify users in a pseudonymised 
Twitter graph. 

On the other hand, \emph{active attacks} rely on the ability to alter 
the structure of the social graph, in such a way that 
the unique structural properties 
allowing to re-identify the victims after publication are guaranteed 
to hold, and to be known by the adversary.   
As we discussed previously, the active attack methodology was introduced 
by Backstrom et al. in~\cite{BDK2007}. They proposed to use sybil nodes 
to create re-identifiable patterns for the victims, in the form of fingerprints 
defined by sybil-to-victim edges. Under this strategy, they proposed two attacks, 
the walk-based attack and the cut-based attack. The difference between both 
attacks lies in the structure given to the sybil subgraph for facilitating 
its retrieval after publication. In the walk-based attack, a long path linking 
all the sybil nodes in a predefined order is created, with remaining inter-sybil 
edges randomly generated. In the cut-based attack, a subset of the sybil nodes 
are guaranteed to be the only cut vertices linking the sybil subgraph 
and the rest of the graph. Interestingly, Backstrom et al. also study 
a passive version of these attacks, where fingerprints are used as identifying 
information, but no sybil nodes are inserted. Instead, they model the situation 
where legitimate users turn rogue and collude to share 
their neighbourhood information in order to retrieve their own weakly induced  
subgraph and re-identify some of their remaining neighbours. 
However, the final conclusion of this study is that the active attack 
is more capable because sybil nodes can better guarantee to create 
a uniquely retrievable subgraph and unique fingerprints. Additionally, a hybrid 
strategy was proposed by Peng et al. \cite{Peng2012,PLZW2014}. 
This hybrid attack is composed of two stages. First, a small-scale active attack 
is used to re-identify an initial set of victims, and then a passive attack 
is used to iteratively enlarge the set of re-identified victims with neighbours 
of previously re-identified victims. Because of the order in which the active 
and the passive phases are executed, the success of the initial active attack 
is critical to the entire attack. 

A large number of anonymisation methods have been proposed for privacy-preserving 
publication of social graphs, which can be divided into three categories: 
those that produce a perturbed version of the original graph 
\cite{LT2008,ZP2008,ZCO2009,Cheng2010,Lu2012,UMGA,Chester2013,Wang2014,Ma2015,Salas2015,Rousseau2017,Casas-Roma2017}, 
those that generate a new synthetic graph sharing some statistical properties 
with the original graph \cite{HMJTW2008,Mittal2013,LinkMirage}, 
and those that output some aggregate statistic of the graph 
without releasing the graph itself, e.g. differentially private degree 
correlation statistics \cite{Sala2011}. 
Active attacks, both the original formulation and the robust version presented 
in this paper, are relevant to the first type of releases. In this context, 
a number of methods have been proposed aiming 
to transform the graph into a new one satisfying some anonymity property based 
on the notion of $k$-anonymity \cite{S2001,Sweeney2002}. Examples of this type 
of anonymity properties for passive attacks are $k$-degree anonymity \cite{LT2008}, 
$k$-neighbourhood anonymity \cite{ZP2008} and $k$-automorphism \cite{ZCO2009}. 
For the case of active attacks, the notions of $(k,\ell)$-anonymity 
was introduced by Trujillo and Yero in~\cite{TY2016}. A $(k,\ell)$-anonymous 
graph guarantees that an active attacker with the ability to insert 
up to $\ell$ sybil nodes in the network will still be unable to distinguish 
any user from at least other $k-1$ users, in terms of their distances 
to the sybil nodes. Several relaxations of the notion of $(k,\ell)$-anonymity 
were introduced in in~\cite{MRT2018}. The notion of $(k,\ell)$-adjacency anonymity 
accounts for the unlikelihood of the adversary to know all distances 
in the original graph, whereas $(k,\Gamma_{G,\ell})$-anonymity models 
the protection of the victims only from vertex subsets with a sufficiently high 
re-identification probability and $(k,\Gamma_{G,\ell})$-adjacency anonymity 
combines both criteria. 

Anonymisation methods based on the notions of $(k,\ell)$-anonymity, 
$(k,\Gamma_{G,\ell})$-anonymity and $(k,\Gamma_{G,\ell})$-adjacency anonymity 
were introduced in~\cite{MauwTrujilloXuan2016,MauwRamirezTrujillo2016,MRT2018}. 
As we discussed above, despite the fact that these methods only give 
a theoretical privacy guarantee 
against adversaries with the capability of introducing a small number 
of sybil nodes, empirical results show that they are in fact capable 
of thwarting attacks leveraging larger numbers of sybil nodes.   
These results are in line with the observation that random perturbations also 
thwart active attacks in their original formulation \cite{NS2009,Ji2015}. 
This is a result of the fact that, originally, active attacks attempt 
to exactly retrieve the sybil subgraph and match the fingerprints. 

In the context of obfuscation methods, which aim to publish a new version 
of the social graph with randomly added perturbations, Xue et al.~\cite{XueKRKP12} 
assess the possibility of the attacker leveraging the knowledge about 
the noise generation to launch what they call a \emph{probabilistic} attack. 
In their work, Xue et al. provided accurate estimators for several 
graph parameters in the noisy graphs, to support the claim that 
useful computations can be conducted on the graphs after adding noise. 
Among these estimators, they included one for the degree sequence of the graph. 
Then, noting that an active attacker can indeed profit from this estimator 
to strengthen the walk-based attack, they show that after increasing 
the perturbation by a sufficiently small amount this attack also fails. 
Although the probabilistic attack presented in \cite{XueKRKP12} 
features some limited level of noise resilience, it is not usable 
as a general strategy, because it requires the noise to follow 
a specific distribution and the parameters of this distribution 
to be known by the adversary. 
Our definition of robust attack makes no assumptions about the type 
of perturbation applied to the graph.   

Finally, we point out that the active attack strategy shares some similarities  
with graph watermarking methods, e.g. \cite{Collberg2003,Zhao2015,Eppstein2016}. 
The purpose of graph watermarking is to release a graph containing 
embedded instances of a small subgraph, the \emph{watermark}, 
that can be easily retrieved 
by the graph publisher, while remaining imperceptible to others 
and being hard to remove or distort. Note that the goals of the graph owner 
and the adversary are to some extent inverted in graph watermarking, 
with respect to active attacks. Moreover, since the graph owner knows 
the entire graph, he can profit from this knowledge for building the watermark. 
During the sybil subgraph creation phase of an active attack, only a partial 
view of the social graph is available to the attacker. 

\section{Conclusions}\label{sec-conclusions}

In this study, we have re-assessed the capabilities of active attackers 
in the setting of privacy-preserving publication of social graphs. 
In particular, we have given definitions of robustness 
for different stages of the active attack strategy and have shown, 
both theoretically and empirically, scenarios under which these notions 
of robustness lead to considerably more successful attacks. 
One particular criticism found in the literature, that of active attacks lacking 
resiliency even to a small number of changes in the network, 
has been shown in this paper not to be an inherent problem of the active 
attack strategy itself, but rather of specific instances of it. 
In light of the results presented here, we argue that active attacks 
should receive more attention by the privacy-preserving social graph 
publication community. In particular, existing privacy measures 
and anonymisation algorithms should be revised, and new ones should be devised, 
to account for the capabilities of robust active attackers.   

\vspace*{.2cm}
\noindent \textbf{Acknowledgements:} The work reported in this paper 
received funding from Luxembourg's Fonds National de la Recherche (FNR), 
via grant C17/IS/11685812 (PrivDA).


\begin{thebibliography}{10}

\bibitem{BDK2007}
Lars Backstrom, Cynthia Dwork, and Jon Kleinberg.
\newblock Wherefore art thou r3579x?: anonymized social networks, hidden
  patterns, and structural steganography.
\newblock In {\em Proceedings of the 16th international conference on World
  Wide Web}, WWW '07, pages 181--190, New York, NY, USA, 2007. ACM.

\bibitem{Bunke2000}
Horst Bunke.
\newblock Recent developments in graph matching.
\newblock In {\em icpr}, page 2117. IEEE, 2000.

\bibitem{UMGA}
Jordi Casas-Roma, Jordi Herrera-Joancomart{\'\i}, and Vicen{\c{c}} Torra.
\newblock An algorithm for k-degree anonymity on large networks.
\newblock In {\em Proceedings of the 2013 IEEE/ACM International Conference on
  Advances in Social Networks Analysis and Mining}, pages 671--675. ACM, 2013.

\bibitem{Casas-Roma2017}
Jordi Casas-Roma, Jordi Herrera-Joancomart{\'\i}, and Vicen{\c{c}} Torra.
\newblock k-degree anonymity and edge selection: improving data utility in
  large networks.
\newblock {\em Knowledge and Information Systems}, 50(2):447--474, 2017.

\bibitem{Cheng2010}
James Cheng, Ada Wai-chee Fu, and Jia Liu.
\newblock K-isomorphism: privacy preserving network publication against
  structural attacks.
\newblock In {\em Proceedings of the 2010 ACM SIGMOD International Conference
  on Management of data}, pages 459--470, 2010.

\bibitem{Chester2013}
Sean Chester, Bruce~M Kapron, Ganesh Ramesh, Gautam Srivastava, Alex Thomo, and
  S~Venkatesh.
\newblock Why waldo befriended the dummy? k-anonymization of social networks
  with pseudo-nodes.
\newblock {\em Social Network Analysis and Mining}, 3(3):381--399, 2013.

\bibitem{Collberg2003}
Christian Collberg, Stephen Kobourov, Edward Carter, and Clark Thomborson.
\newblock Error-correcting graphs for software watermarking.
\newblock In {\em Proceedings of the 29th Workshop on Graph Theoretic Concepts
  in Computer Science}, pages 156--167. Springer, 2003.

\bibitem{Eppstein2016}
David Eppstein, Michael~T Goodrich, Jenny Lam, Nil Mamano, Michael
  Mitzenmacher, and Manuel Torres.
\newblock Models and algorithms for graph watermarking.
\newblock In {\em International Conference on Information Security}, pages
  283--301. Springer, 2016.

\bibitem{Fober2013}
Thomas Fober, Gerhard Klebe, and Eyke H{\"u}llermeier.
\newblock Local clique merging: An extension of the maximum common subgraph
  measure with applications in structural bioinformatics.
\newblock In {\em Algorithms from and for Nature and Life}, pages 279--286.
  Springer, 2013.

\bibitem{HMJTW2008}
Michael Hay, Gerome Miklau, David Jensen, Don Towsley, and Philipp Weis.
\newblock Resisting structural re-identification in anonymized social networks.
\newblock {\em Proc. VLDB Endow.}, 1(1):102--114, August 2008.

\bibitem{Ji2015}
Shouling Ji, Weiqing Li, Prateek Mittal, Xin Hu, and Raheem~A Beyah.
\newblock Secgraph: A uniform and open-source evaluation system for graph data
  anonymization and de-anonymization.
\newblock In {\em USENIX Security Symposium}, pages 303--318, 2015.

\bibitem{LinkMirage}
Changchang Liu and Prateek Mittal.
\newblock Linkmirage: Enabling privacy-preserving analytics on social
  relationships.
\newblock In {\em Procs. of NDSS'16}, 2016.

\bibitem{LT2008}
Kun Liu and Evimaria Terzi.
\newblock Towards identity anonymization on graphs.
\newblock In {\em Proceedings of the 2008 ACM SIGMOD International Conference
  on Management of Data}, SIGMOD '08, pages 93--106, New York, NY, USA, 2008.
  ACM.

\bibitem{Lu2012}
Xuesong Lu, Yi~Song, and St{\'e}phane Bressan.
\newblock Fast identity anonymization on graphs.
\newblock In {\em International Conference on Database and Expert Systems
  Applications}, pages 281--295. Springer, 2012.

\bibitem{Ma2015}
Tinghuai Ma, Yuliang Zhang, Jie Cao, Jian Shen, Meili Tang, Yuan Tian, Abdullah
  Al-Dhelaan, and Mznah Al-Rodhaan.
\newblock Kdvem: a k-degree anonymity with vertex and edge modification
  algorithm.
\newblock {\em Computing}, 97(12):1165--1184, 2015.

\bibitem{Mallek2015}
Sabrine Mallek, Imen Boukhris, and Zied Elouedi.
\newblock Community detection for graph-based similarity: application to
  protein binding pockets classification.
\newblock {\em Pattern Recognition Letters}, 62:49--54, 2015.

\bibitem{MauwRamirezTrujillo2016}
Sjouke Mauw, Yunior Ram\'{i}rez-Cruz, and Rolando Trujillo-Rasua.
\newblock Anonymising social graphs in the presence of active attackers.
\newblock {\em Transactions on Data Privacy}, 11(2):169--198, 2018.

\bibitem{MRT2018}
Sjouke Mauw, Yunior Ram\'{i}rez-Cruz, and Rolando Trujillo-Rasua.
\newblock Conditional adjacency anonymity in social graphs under active
  attacks.
\newblock {\em Knowledge and Information Systems}, to appear, 2018.

\bibitem{MauwTrujilloXuan2016}
Sjouke Mauw, Rolando Trujillo-Rasua, and Bochuan Xuan.
\newblock Counteracting active attacks in social network graphs.
\newblock In {\em Proceedings of DBSec'16}, volume 9766 of {\em Lecture Notes
  in Computer Science}, pages 233--248, 2016.

\bibitem{Mittal2013}
Prateek Mittal, Charalampos Papamanthou, and Dawn Song.
\newblock Preserving link privacy in social network based systems.
\newblock In {\em NDSS 2013}, 2013.

\bibitem{NS2009}
Arvind Narayanan and Vitaly Shmatikov.
\newblock De-anonymizing social networks.
\newblock In {\em Procs. of the 30th IEEE Symposium on Security and Privacy},
  pages 173--187, 2009.

\bibitem{Peng2012}
Wei Peng, Feng Li, Xukai Zou, and Jie Wu.
\newblock Seed and grow: An attack against anonymized social networks.
\newblock In {\em 9th Annual {IEEE} Communications Society Conference on
  Sensor, Mesh and Ad Hoc Communications and Networks, {SECON} 2012, Seoul,
  Korea (South), June 18-21, 2012}, pages 587--595, 2012.

\bibitem{PLZW2014}
Wei Peng, Feng Li, Xukai Zou, and Jie Wu.
\newblock A two-stage deanonymization attack against anonymized social
  networks.
\newblock {\em IEEE Transactions on Computers}, 63(2):290--303, Feb 2014.

\bibitem{Rousseau2017}
Fran{\c{c}}ois Rousseau, Jordi Casas-Roma, and Michalis Vazirgiannis.
\newblock Community-preserving anonymization of graphs.
\newblock {\em Knowledge and Information Systems}, 54(2):315--343, 2017.

\bibitem{Sala2011}
Alessandra Sala, Xiaohan Zhao, Christo Wilson, Haitao Zheng, and Ben~Y Zhao.
\newblock Sharing graphs using differentially private graph models.
\newblock In {\em Procs. of the 2011 ACM SIGCOMM Conf. on Internet
  Measurement}, pages 81--98, 2011.

\bibitem{Salas2015}
Juli{\'a}n Salas and Vicen{\c{c}} Torra.
\newblock Graphic sequences, distances and k-degree anonymity.
\newblock {\em Discrete Applied Mathematics}, 188:25--31, 2015.

\bibitem{S2001}
Pierangela Samarati.
\newblock Protecting respondents' identities in microdata release.
\newblock {\em IEEE Trans. on Knowl. and Data Eng.}, 13(6):1010--1027, 2001.

\bibitem{SF1983}
Alberto Sanfeliu and King{-}Sun Fu.
\newblock A distance measure between attributed relational graphs for pattern
  recognition.
\newblock {\em {IEEE} Trans. Systems, Man, and Cybernetics}, 13(3):353--362,
  1983.

\bibitem{Sweeney2002}
Latanya Sweeney.
\newblock k-anonymity: A model for protecting privacy.
\newblock {\em International Journal of Uncertainty, Fuzziness and
  Knowledge-Based Systems}, 10(5):557--570, 2002.

\bibitem{TY2016}
Rolando Trujillo{-}Rasua and Ismael~Gonz{\'{a}}lez Yero.
\newblock k-metric antidimension: {A} privacy measure for social graphs.
\newblock {\em Information Sciences}, 328:403--417, 2016.

\bibitem{VBCG_HPCS14}
S.~Varrette, P.~Bouvry, H.~Cartiaux, and F.~Georgatos.
\newblock Management of an academic {HPC} cluster: The {UL} experience.
\newblock In {\em Proceedings of the 2014 International Conference on High
  Performance Computing \& Simulation}, pages 959--967, Bologna, Italy, 2014.

\bibitem{Wang2014}
Yazhe Wang, Long Xie, Baihua Zheng, and Ken~CK Lee.
\newblock High utility k-anonymization for social network publishing.
\newblock {\em Knowledge and information systems}, 41(3):697--725, 2014.

\bibitem{XueKRKP12}
Mingqiang Xue, Panagiotis Karras, Chedy Ra{\"{\i}}ssi, Panos Kalnis, and
  Hung~Keng Pung.
\newblock Delineating social network data anonymization via random edge
  perturbation.
\newblock In {\em 21st {ACM} International Conference on Information and
  Knowledge Management, CIKM'12, Maui, HI, USA, October 29 - November 02,
  2012}, pages 475--484, 2012.

\bibitem{Yu2008}
Haifeng Yu, Phillip~B Gibbons, Michael Kaminsky, and Feng Xiao.
\newblock Sybillimit: A near-optimal social network defense against sybil
  attacks.
\newblock In {\em Proceedings of the 2008 IEEE Symposium on Security and
  Privacy}, pages 3--17, Oakland, CA, USA, 2008.

\bibitem{Yu2006}
Haifeng Yu, Michael Kaminsky, Phillip~B Gibbons, and Abraham Flaxman.
\newblock Sybilguard: defending against sybil attacks via social networks.
\newblock In {\em Proceedings of the 2006 Conference on Applications,
  Technologies, Architectures, and Protocols for Computer Communications},
  pages 267--278, Pisa, Italy, 2006.

\bibitem{Zhao2015}
Xiaohan Zhao, Qingyun Liu, Haitao Zheng, and Ben~Y Zhao.
\newblock Towards graph watermarks.
\newblock In {\em Proceedings of the 2015 ACM on Conference on Online Social
  Networks}, pages 101--112. ACM, 2015.

\bibitem{ZP2008}
Bin Zhou and Jian Pei.
\newblock Preserving privacy in social networks against neighborhood attacks.
\newblock In {\em Proceedings of the 2008 IEEE 24th International Conference on
  Data Engineering}, ICDE '08, pages 506--515, Washington, DC, USA, 2008. IEEE
  Computer Society.

\bibitem{ZCO2009}
Lei Zou, Lei Chen, and M.~Tamer \"{O}zsu.
\newblock K-automorphism: A general framework for privacy preserving network
  publication.
\newblock {\em Proc. VLDB Endow.}, 2(1):946--957, August 2009.

\end{thebibliography}
\end{document}